\newif\iflncs\lncstrue
\newif\ifmai\maifalse
\documentclass{lmcs} 
\pdfoutput=1

\usepackage{lastpage}
\lmcsdoi{20}{2}{15}
\lmcsheading{}{\pageref{LastPage}}{}{}%
{Mar.~22,~2023}{Jun.~18,~2024}{}

\usepackage{color}
\usepackage{hyperref}
\usepackage[english]{babel}
\usepackage{bbm}
\usepackage{amsthm}
\usepackage[T1]{fontenc}
\usepackage[utf8]{inputenc} 
\usepackage{algorithm}
\usepackage{svg}
\usepackage{algpseudocode}
\usepackage{listings}
\usepackage{ifthen}

\usepackage{chngcntr}
\usepackage{apptools}

\usepackage{booktabs}
\usepackage{multirow}
\usepackage{amsthm}


\newenvironment{proof_of}[1]{\begin{proof}[Proof of #1.]}{\end{proof}}

\usepackage{amssymb,amsmath}
\usepackage{graphicx}
\usepackage{pgf}
\usepackage{tikz}
\usepackage{wrapfig}

\newif\ifmai\maifalse

\newcounter{ncomm}%

\newcommand{\BG}{$(F,G)$} 
\newcommand{\BBGG}{well-behaved} 

\usepackage{fontawesome5}

\newcommand{\revOne}[1]{{\color{black} #1}}
\newcommand{\revTwo}[1]{{\color{black} #1}}

\newcommand{\drifts}{{\cal D}}
\newcommand{\incond}{\rho}
\newcommand{\der}{\delta}
\newcommand{\pr}{\pi}
\newcommand{\Str}{\Sigma}

\newcommand{\had}{\odot}

\newcommand{\E}{\mathcal{E}} 
\newcommand{\B}{\mathcal{B}} 
\newcommand{\Cn}{\mathcal{C}} 
\newcommand{\An}{\mathcal{A}} 
\newcommand{\Lie}{\mathrm{L}} 
\newcommand{\Eqs}{\mathcal{D}}  
\newcommand{\reals}{\mathbb{R}}  
\newcommand{\C}{\mathbb{C}}  
\newcommand{\ode}{{\sc ODE}}
\newcommand{\sde}{{\sc SDE}}

\newcommand{\defii}{:=}   

\newcommand{\convo}{\times} 
\newcommand{\K}{\mathbb{K}} 
\newcommand{\N}{\mathbb{N}} 

\newcommand{\Pol}{\mathcal{P}}

\makeatletter 

\newcommand*\Neginternal[3]{\mathpalette\Neg@{{#1}{#2}{#3}}}
\newcommand*\Neg@[2]{\Neg@@{#1}#2}
\newcommand*\Neg@@[4]{%
	\mathrel{\ooalign{%
			$\m@th#1#4$\cr
			\hidewidth$\m@th#3{#1}\mkern\muexpr#2*2$\hidewidth\cr
	}}%
}

\newcommand*\negslash[1]{\m@th#1\not\mathrel{\phantom{=}}}
\newcommand*\snegslash[1]{\rotatebox[origin=c]{60}{$\m@th#1-$}}
\newcommand*\ssnegslash[1]{\rotatebox[origin=c]{60}{$\m@th#1{\dabar@}\mkern-7mu{\dabar@}$}}
\newcommand*\sssnegslash[1]{\rotatebox[origin=c]{60}{$\m@th#1\dabar@$}}
\makeatother  

\newcommand{\altfrac}[2]{\ifmmode\def\tmp{$}\else\def\tmp{}\fi\mbox{%
		{\raisebox{.24\ht\strutbox}{\tmp#1\tmp}}%
		\kern-2.2pt\scalebox{1.6}[1.5]{/}\kern-1.8pt%
		{\tmp#2\tmp}%
}}

\newcommand{\ccc}{\mathrm{\mathbf{c}}}
\newcommand{\ov}[1]{{\pmb{\mathrm{#1}}}}

\usepackage{chngcntr}
\usepackage{apptools}

\begin{document}
		
\title{An implicit function theorem  for the stream calculus}

\author[M.~Boreale]{Michele Boreale}[a]
\author[L.~Collodi]{Luisa Collodi}[a]
\author[D.~Gorla]{Daniele Gorla}[b]

\address{Università di Firenze, Italy.}	
\email{michele.boreale@unifi.it, luisa.collodi@unifi.it}  

\address{‘‘Sapienza'' Università di Roma, Italy.}	
\email{gorla@di.uniroma1.it}  

\begin{abstract}
	In the context of the stream calculus, we present an Implicit Function Theorem  (IFT) for polynomial systems,  and   discuss its relations with the classical IFT from calculus. In particular, we demonstrate the advantages of the stream IFT from a computational point of view, and provide a few example applications  where its use turns out to be valuable.
\end{abstract}
		
		\maketitle

\section{Introduction}\label{sec:intro}
In theoretical computer science, the last two decades have seen an increasing interest in the concept of \emph{stream} and in the related proof techniques, collectively designated as the \emph{stream calculus} \cite{Rut01,Rut03,Rut05}.
A \emph{stream} $\sigma=(r_0,r_1,...)$ is an infinite sequence of elements (coefficients) $r_i$   drawn from
\revOne{any set; the stream calculus  requires that this set be endowed with some algebraic structure, such as a field.}
Therefore, as a concrete mathematical object,  a stream is just the same   as  a \emph{formal power  series} considered in combinatorics and other fields of mathematics. The use of a different terminology here is motivated by the fact that, with streams, the basic computational device is that of \emph{stream derivative}, as opposed to the ordinary derivative from calculus   considered in formal power series. The stream derivative  $\sigma'$    is obtained  by simply removing the first element $r_0$ from   $\sigma$. \revOne{The ordinary derivative of a formal power series, on the other hand,  does not enjoy such a simple formulation in terms of stream manipulation. The simplicity of stream derivative also leads to   computational advantages, as  discussed further below.
	Algebraically, the stream derivative enjoys a nice relation with the operation of \emph{convolution} $\convo$ (one of the possible notions of product for streams, to be introduced in Section \ref{sec:background}), as expressed by the so-called \emph{fundamental theorem} of the stream calculus:
	$$\sigma=\sigma(0)+X\convo \sigma'\,.$$
	Here, $\sigma(0)$ is taken as an abbreviation of the stream $(\sigma(0),0,0,...)$, while $X =(0,1,0,0,...)$. As can be intuited from this equation, multiplying a stream ($\sigma'$ in this case) by $X$ has the effect of shifting one position to the right the stream's coefficients.}

A powerful and elegant proof technique for streams is  \emph{coinduction} \cite{Sangio}, whose step-by-step flavour  naturally agrees  with the above mentioned features of streams, in particular stream derivative. Moreover,
an important specification and computational device is represented by \emph{stream differential equations} (\sde s, \cite{HKR17}),  the analog of   \emph{ordinary} differential equations (\ode s)  \revOne{for functions and formal power series \cite{flajo,Stanley}}. It is this toolkit of mechanisms and proof techniques that one collectively designates as the  {stream calculus} \cite{Rut01,Rut03,Rut05}.
One point of strength of the stream calculus is that it provides simple, direct and unified reasoning techniques, that can be applied to a variety of systems that involve the treatment of sequences.  A distinguished feature of proofs conducted within the stream calculus is that issues related to convergence (of sequences, functions etc.) basically never enter the picture. As an example,  the stream calculus has been proved valuable in providing a coinductive account of analytic functions and Laplace transform \cite{Escardo}, in solving difference and differential equations \cite{Bor19,Bor20,Rut03,Rut05}, as well as in formalizing several versions of signal flow graphs \cite{BSZ15,BSZ17,Rut03a,Rut05a}.
In Section \ref{sec:background}, we provide a quick overview of the basic definitions and features of the stream calculus.

The main goal and  contribution of the present work is to     add yet another tool to the stream calculus:  an Implicit Function Theorem (IFT) for systems of  stream polynomial  equations. Indeed, while \sde s represent a powerful computational device, depending on the problem at hand, streams may be  more naturally expressed in an algebraic fashion, that is as the (unique) solution of systems of polynomial equations. In analogy with the classical IFT from calculus \cite{Krantz,Rudin,IFT}, our main result provides sufficient syntactic  conditions under which a system of polynomial equations has a unique stream solution. Moreover, the theorem also provides an equivalent system of \sde s, that is useful to actually generate the stream solution. \revOne{It is here that the computational advantage of stream derivatives, as opposed to  ordinary ones,  clearly shows up.}

\revOne{In the classical IFT \cite[Th.9.28]{Rudin},
one considers a system  of  equations in the variables $( x,\ov y)$, say
$\ov F( x,\ov y)=0$; for simplicity, here we assume that $x$ is a scalar, while $\ov y$ can be a vector. 
The IFT  gives sufficient conditions under which,
for any fixed point  $(  x_0,\ov y_0)$ satisfying $\ov F( x_0,\ov y_0)=0$,    a neighbourhood around $x_0$ and a map $f$ on that neighbourhood exist fulfilling $f(x_0)=\ov y_0$ and  $\ov F(x, f(x)) = 0$.
Otherwise said, $\ov F(  x,\ov y)=0$ implicitly defines a function  $f$ (hence the name of the theorem) such that $f(x_0)=\ov y_0$ and $ \ov F(x,f(x))$ is identically 0 in a neighborhood of $x_0$, 
}
The required sufficient condition is that the Jacobian matrix  (the matrix of partial derivatives) of $\ov F$ with respect to $\ov y$ be nonsingular when evaluated at $(  x_0,\ov y_0)$.
The theorem also gives a system of   \ode s whose solution is the function $  f(  x)$. Although the  \ode\  system  will be in general impossible to solve analytically, it can be used to compute a truncated Taylor series of $  f(   x)$  to the desired degree of approximation.   

In Section \ref{sec:general}, in the setting of the stream calculus and of polynomial equations, we obtain a version of the IFT whose form  closely resembles  the classical one (Theorem \ref{th:ift}). The major difference is that the stream version relies, of course, on stream derivatives, and on a corresponding    notion of stream Jacobian. In particular, the system of \ode s that defines the solution is here   replaced by a system of \sde s.
A crucial step towards proving the result is devising a stream version of  the  chain rule from calculus, whereby one can express the derivative of a function $\ov F(x,y_1(x),...,y_n(x))$  with respect to $x$  in terms of the partial derivatives of $\ov F$ with respect to $y_i$ and the ordinary derivative of the $y_i$ with respect to $x$.

In Section \ref{sec:classical}, beyond the formal similarity, we discuss the precise mathematical relation of the stream IFT with the classical IFT (Theorem \ref{th:class}). We show that the two theorems can be applied precisely under the same   assumptions on the classical Jacobian of $\ov F(  x,\ov y)$. Moreover, the sequence of Taylor coefficients of the function defined by the classical IFT coincides with the solution identified by the stream IFT. Therefore, one has two alternative methods to compute the (stream) solution.
Despite this close relationship,  the stream version of the theorem  is conceptually and computationally very different from the classical one; the computational aspects will be further discussed below. In Section \ref{sec:classical}, we also discuss the relation of the stream IFT  with {algebraic series} as considered in enumerative combinatorics \cite{flajo,Stanley}.


As an extended example of application of the stream IFT, in Section \ref{sec:casestudy}  we apply the result to the problem of enumerating   \emph{three-colored trees} \cite[Sect.4, Example 14]{flajo}, a typical class of combinatorial objects  that are most naturally described by  algebraic equations.

In Section \ref{sec:computational} we  discuss the computational aspects of the stream IFT. We first outline an efficient method to calculate the coefficients of the stream solution up to a prescribed order, based on the \sde\ system   provided by the theorem. Then  we   offer an empirical comparison between two methods to compute the stream solution: the above mentioned method based on the stream IFT, and  the method based on the \ode s provided by the classical IFT.
This comparison clearly shows the computational benefits of first method (stream IFT)  over the second one (classical IFT)  in terms of running time.
An important point is that, when applied to polynomials, the syntactic size of stream derivatives is approximately \emph{half} the size of ordinary (classical) derivatives. 

We conclude the paper in Section \ref{sec:concl} with a brief discussion on  possible directions for future research.

\paragraph*{Related work}
The stream calculus in the form considered here has been introduced by Rutten in a series of works, in particular \cite{Rut01} and \cite{Rut03}.
In \cite{Rut01}, streams and operators on streams are introduced via coinductive definitions and behavioural differential equations, later called stream differential equations, involving initial conditions and derivatives of streams. Several applications are also presented to:  difference equations, analytical differential equations, and some problems from discrete mathematics and combinatorics. In \cite{Rut03} streams, automata, languages and formal power series are studied in terms of the notions of coalgebra homomorphisms and bisimulation.

A recent development of the stream calculus that is related to the present work is \cite{BG21}, where the authors introduce a   polynomial format for \sde s\  and an algorithm  to automatically check  polynomial equations, with respect to a  \emph{generic} notion of  {product} for streams  satisfying certain conditions. These results can be applied to {\em convolution} and \emph{shuffle} products, among the others.

In     formal language theory, context-free grammars can be viewed as   instances of polynomial systems: see   \cite{KS}. A   coinductive treatment of this type of systems is found in Winter's work  \cite[Ch.3]{Win14}.   Note that, on one hand,  the polynomial format we consider here is significantly more expressive than context-free grammars, as we can deal with such  equations as $x^2+y^2-1 =0$ (see Example \ref{ex:guarded} in Section \ref{sec:general}) that are outside the context-free format. On the other hand, here  we confine to \emph{univariate} streams, which can be regarded as weighted languages on the alphabet $\{x\}$, whereas in language theory alphabets of any finite size can be considered. How to extend the present results to multivariate streams is a \revOne{challenging direction for future research}.

In enumerative combinatorics \cite{flajo,Stanley}, formal power series defined via polynomial equations are named \emph{algebraic series}.  \cite[Sect.4]{flajo} discusses several aspects of algebraic series, including several methods of reduction, involving the theory of resultants and Groebner bases. We compare our approach to algebraic series in Section \ref{sec:general}, Remark \ref{rem:algebraic}.

\newcommand{\Y}{\mathcal{Y}}
\section{Background}\label{sec:background}
\subsection{Streams}
\revOne{Let $(\K, 0, 1, +, \cdot)$ be a field.}
We let  $\Str\langle \K\rangle\defii \K^\omega$, ranged over by $\sigma,\tau,...$, denote the set of \emph{streams}, \revOne{that are infinite sequences  of elements from $\K$}: $\sigma=(r_0,r_1,r_2,...)$ with $r_i\in \K$. Often $\K$ is understood from the context and we shall simply write $\Str$ rather than $\Str\langle \K\rangle$.
When convenient, we shall explicitly consider a stream $\sigma$ as a function from $\mathbb N$ to $\K$
and, e.g., write $\sigma(i)$ to denote the $i$-th element of $\sigma$.
By slightly overloading the notation, and when   the context is sufficient to disambiguate, the stream $(r,0,0,...)$ ($r\in \K$) will be simply denoted  by  $r$, while the stream $(0,1,0,0,...)$ will be denoted by $X$;  see \cite{Rut03} for motivations behind these notations.
Given two streams $\sigma$ and $\tau$, we define the streams $\sigma+\tau$ (sum) and $\sigma\times\tau$ (convolution product)
by
\begin{align}\label{eq:sumconv}
(\sigma+\tau)(i)&\defii \sigma(i)+\tau(i) &  (\sigma\convo \tau)(i)& \defii \sum_{0\leq j\leq i}\sigma(j)\cdot\tau(i-j)
\end{align}
for each $i\geq 0$, where the $+$ and $\cdot$ on the right-hand sides above denote    sum and product in $\K$, respectively. Sum enjoys the usual commutativity and associativity properties, and has the stream $0=(0,0,...)$ as an identity.

Convolution product is commutative,  associative,   has   $1=(1,0,0,...)$ as an identity, and distributes over $+$;
\revOne{
differently from, e.g., \cite{BG21,HKR17}, here we only consider the convolution product.
}

Multiplication of $\sigma=(r_0,r_1,...)$  by a scalar $r\in \K$, denoted $r\sigma=(r\,r_0,r\, r_1,...)$, is also defined and makes $(\Sigma,+,0)$ a vector space over $\K$. Therefore, $(\Sigma, +,\convo,0,1)$ forms a \revTwo{(associative)} $\K$-algebra. 
We also record the following   facts for future use:
$
X \convo \sigma   = (0,  \sigma(0),\sigma(1),...)
$ and
$
r \,\convo\, \sigma  = (r\,\sigma(0),r\, \sigma(1),...) 
$, where $r\in \K$.
\ In view of the second equation above, $r\,\convo\,\sigma$ coincides with
$r\sigma$.

For each $\sigma$, we let its \emph{derivative}  $\sigma'$ be the stream defined by $\sigma'(i)=\sigma(i+1)$ for  each $i\geq 0$. In other words, $\sigma'$ is obtained from $\sigma$ by removing the first element $\sigma(0)$.
The equality $X \convo \sigma  = (0,  \sigma(0),\sigma(1),...)$   above leads to the so called fundamental theorem of the stream calculus, whereby for each $\sigma\in \Sigma$
\begin{equation}\label{eq:ftsc}
\sigma  =\sigma(0)+X\times \sigma'\,.
\end{equation}
Every stream $\sigma$ such that $\sigma(0)\neq 0$ has a unique   inverse with respect to convolution, denoted $\sigma^{-1}$, that satisfies the equations:
\begin{align}\label{eq:inv}
(\sigma^{-1})'&=-\sigma(0)^{-1}\cdot (\sigma'\convo \sigma^{-1}) & (\sigma^{-1})(0)&=\sigma(0)^{-1}\,.
\end{align}

\subsection{Polynomial stream differential equations}
Let us fix a finite, non empty set of symbols or \emph{variables}  $\Y = \{y_1,\ldots,y_n\}$ and a \revTwo{distinguished} variable $x\notin \Y$. Notationally, when fixed an order on such variables, we use the notation $\ov y:=(y_1,...,y_n)$.
We fix a generic field $\K$ of characteristic 0;   $\K=\reals$ and $\K=\C$ will be typical choices.
We let $\Pol\defii\K[x,y_1,...,y_n]$, ranged over by $p,q,...$, be the set of polynomials with coefficients in $\K$ and indeterminates in $\{x\}\cup \Y$. 
As usual, we shall denote polynomials as formal finite sums of distinct monomials  with coefficients in $\K$: $p=\sum_{i\in I}r_i m_i$, for $r_i\in \K$ and $m_i$ monomials over $\{x\}\cup\Y$.
For the sake of uniform notation, we shall sometimes let  $y_0$ denote $x$, so we can write a generic monomial in $\Pol$ as $y_0^{k_0}\cdots y_n^{k_n}$, for $k_i \in \N$ for every $i$.
By slight abuse of notation, we shall write the zero polynomial and the empty monomial as $0$ and $1$, respectively.

Over $\Pol$, one   defines the usual operations of sum   $p+q$ and product $p\cdot q$, with 0 and 1 as identities,  and enjoying commutativity, associativity and distributivity, which make $\Pol$ a ring. Multiplication of $p\in \Pol$  by a scalar $r\in \K$, denoted $rp$, is also defined and makes $(\Pol,+,0)$ a vector space over $\K$. Therefore, \revOne{$(\Pol, +,\times, 0, 1)$ as well forms a  free $\K$-{algebra}   with generators $(x, y_1, ..., y_n)$.}
For each $n$-tuple of streams $\ov\sigma=(\sigma_1,...,\sigma_n)$, there is a unique $\K$-algebra homomorphism $\phi_{\ov\sigma}:\Pol\longrightarrow \Str$ such that $\phi_{\ov\sigma}(x)= X$ and $\phi_{\ov\sigma}(y_i)=\sigma_i$ for $i=1,...,n$. For any $p\in \Pol$, we let $p(X,\ov\sigma):=\phi_{\ov\sigma}(p)$, that is the result of substituting the variables $x$ and $\ov y$ in $p$ with the streams $X$ and $\ov\sigma$, respectively.

\begin{defi}[\sde\ \cite{Rut03}]\label{def:sde}
Given a tuple of polynomials $(p_1,...,p_n)\in \Pol^n$ and $\ov r_0=(r_1,...,r_n)\in \K^n$, the corresponding system of 
(polynomial) \emph{stream differential equations (\sde s)} $\Eqs$ and \emph{initial conditions}   are written as follows
\begin{align}\label{eq:ivp}
	\Eqs&=\{y'_1 = p_1,...,y'_n = p_n\} & &\rho=\{y_1(0)=r_1,...,y_n(0)=r_n\}\,.
\end{align}
The pair $(\Eqs,\incond)$ is also said to form  a (polynomial) \sde\
{\em initial value problem}  for the variables $\ov y$. A \emph{solution} of \eqref{eq:ivp} is  a  tuple of streams $\ov\sigma=(\sigma_1,...,\sigma_n)\in\Str^n$ such that $\sigma'_i=p_i(X,\ov\sigma)$ and $\sigma_i(0)=r_i$, for $i=1,...,n$.
\end{defi}

\revOne{
A natural generalization of the above definition are systems of  {\emph{rational}} \sde s, where the right-hand side of each equation is   a  fraction of polynomials. Systems of rational \sde s have indeed the same expressive power as polynomial ones: a version of this  (well-known) result will be explicitly formulated  in Section \ref{sec:general} (see Lemma \ref{lemma:rational}).
}
%
%
For a proof of the following theorem (in a more general context), see e.g.   \cite{BG21,HKR17}.

\begin{thm}[existence and uniqueness of solutions]\label{thm:main}
Every polynomial \sde\ initial value problem of the form  \eqref{eq:ivp} has a unique solution.
\end{thm}

\begin{rem}[stream coefficients computation]\label{rem:comp}
We record for future use that a \sde\ initial value problem $(\Eqs,\rho)$  like \eqref{eq:ivp}  yields a  recurrence relation, hence an algorithm, to compute the coefficients of the solution streams $\sigma_i$. Indeed, denote by $\sigma_{:k}$ the stream that coincides with $\sigma$ when restricted to $\{0,...,k\}$ and is 0 elsewhere. This notation is extended to a tuple $\ov\sigma$ componentwise. Then we have, for each $i=1,...,n$ and   $k\geq 0$:
\begin{align}
	\sigma_i(0)&=y_i(0)\\ \label{eq:recurrence}
	\sigma_i(k+1)&= \sigma'_i(k)\,=\,p_i(X,\ov\sigma)(k)\,=\,p_i(X,\ov\sigma_{:k})(k)
\end{align}
where the last step follows from the fact that the $k$-th coefficient of $p_i(X,\ov\sigma)$ only depends on the first $k$ coefficients of $\ov\sigma$ (see \eqref{eq:sumconv}).
\revOne{
	In the literature, this is referred to as {\em causality} (see \cite{HKR17,Klin11,PR17}, just to cite a few).
}

As an example, consider (here we let $y=y_1$):
$$
y'=y^2 \qquad \qquad y(0)=1
$$
for which we get the recurrence: $\sigma(0)=1$ and $\sigma(k+1)=\sigma^2(k)=\sum_{j=0}^k\sigma(j)\cdot\sigma(k-j)$.  From the computational point of view this is far from  optimal.  \revOne{Indeed, in the case of  a single polynomial equation ($n=1$) like this one, a linear (in $y$) recurrence relation for generating the Taylor coefficients of the solution can always be efficiently built; see \cite{flajo,Stanley}}.  In the case of $n>1$ equations, the situation is more complicated. We defer to Section \ref{sec:computational} further  considerations on the computation of stream coefficients, including details on an effective implementation  of \eqref{eq:recurrence}.
\end{rem}

\ifmai
\begin{proposition}
\label{prop:uniq}
Let $\pr$ be a \BBGG\ \BG-product and $\nu$ be a $\K$-algebra
homomorphism from $(\Pol, +, \cdot\,,0, 1)$ to
$(\Sigma,  + , \pr,  0, 1_\pr)$ that respects the initial value problem $(\drifts,\rho)$.
Then, $\nu$ is a coalgebra morphism from $(\Pol, \der_\pr, o_\incond)$ to
$(\Sigma, (\cdot)', o)$.
\end{proposition}
\fi

\newcommand{\Q}{\mathbb{Q}}

\newcommand\nablah{%
\mathrel{\ooalign{$\nabla$\cr\hfil\scalebox{1}[1]{\rotatebox[origin=c]{-70}{$^\setminus$}\!}\hfil\cr}}%
}
\renewcommand\dh{\eth
}

\section{An implicit function theorem for the stream calculus}
\label{sec:general}

Let $\E\subseteq \Pol$  be a finite, nonempty  set of polynomials \revOne{that we call {\em polynomial system}}. A  \emph{stream   solution} of $\E$   is a tuple of streams $\ov\sigma=(\sigma_1,...,\sigma_n)$ such that $p(X,\ov\sigma)=0$ for each $p\in\E$.
We want to show that, under certain syntactic conditions,  $\E$ has a   unique stream solution, which can be also defined via a polynomial \sde\ initial value problem $(\Eqs,\rho)$. Instrumental  to establish this result is a close stream analog of the well known Implicit Function Theorem (IFT) from calculus. 

Let us  introduce some extra notation on polynomials and streams. 
Beside the variables $x$ and   $\ov y =(y_1,...,y_n)$, we shall consider  a   set of new,  distinct  \emph{initial value indeterminates} $\ov y_0=(y_{01},...,y_{0n})$  and \emph{{primed} indeterminates}  $\ov{y'}=(y'_1,...,y'_n)$. \revOne{As usual, we let $y_0=x$; moreover, by slightly abusing notation, we will let $y_{00}$ denote $0$ (the scalar zero).}
\revTwo{
We will assume a fixed total order on all variables
$(x=)y_0<y_1<\cdots<y_n$ and, for any monomial $m\neq 1$, on the variables in $\ov y$ define $\min(m):=\min\{y\,:\,y \text{ occurs in }m\}$, where the $\min$ is taken according to the fixed total order on variables. In the definition below, we order the individual variables in a monomial according to $<$ before proceeding to differentiation. It will turn out that the chosen total order is semantically immaterial,  see Remark \ref{rem:totalorder} further below. The \emph{total degree}
of a monomial $m$ is just its size, that is the number of occurrences of variables in $m$.   Recall that $\K[x,\ov y_{0},\ov y, \ov y']$ denotes the set of polynomials having with coefficients in $\K$ and indeterminates in $x,\ov y_{0},\ov y, \ov y'$.
}

\begin{defi}[syntactic stream derivative]\label{def:stream_der}
\revTwo{The \emph{syntactic stream derivative} operator $(\cdot)':\Pol\rightarrow \K[x,\ov y_{0},\ov y, \ov y']$ is defined as follows. First,  we define $(\cdot)'$ on monomials by induction on the total degree as follows:
	$$
	\begin{array}{cccl}
		(1)':=0  & (x)':=1 & (y_i)':=y'_i & \text{ ($1\leq i\leq n$)  }\\[5pt]
		\multicolumn{3}{l}{ (y_i\cdot m)' := y'_i\cdot m+y_{0i}\cdot (m)'} &  \multicolumn{1}{l}{\text{ ($0\leq i\leq n$,  $y_i=\min(y_i\cdot m)$ and $m\neq 1$).}}
	\end{array}
	$$}
The operator  $(\cdot)'$  is then  extended to polynomials in $\Pol$ by linearity.
\end{defi}

As an example, $(xy^2_1+y_1y_2)'=y^2_1+y'_1y_2+y_{01}y'_2$.
Note that $p'$ lives in a   polynomial ring $\K[x,\ov y_{0},\ov y, \ov y']$ that includes $\Pol$.   We shall write $p'$ as $p'(x,\ov y_0,\ov y,\ov y')$ when   wanting to make     the  indeterminates that may    occur   in $p'$ explicit. With this notation, it is easy to check that  $(\cdot)'$ commutes with substitution, as stated in the following lemma.

\begin{lem}\label{PIPPO}
For every $p(x,\ov y)$ and $\ov \sigma$, we have that $(p(X,\ov \sigma))'=p'(X,\ov\sigma(0),\ov\sigma,\ov\sigma')$.
\end{lem}
\begin{proof}
For $p$ a  monomial, the proof is by induction on its total degree,  and straightforwardly follows from Definition~\ref{def:stream_der}. The general case when  $p$ is a linear combination of monomials  follows then by linearity of the definition of $(\cdot)'$.
\end{proof}

\revOne{
\begin{rem}\label{rem:totalorder} While the definition of  {syntactic} stream derivative \emph{does} depend on the chosen total order of indeterminates $(y_0,...,y_n)$, Lemma \ref{PIPPO} confirms that this order becomes immaterial when the indeterminates are substituted with streams. In particular, if $(\cdot)'$ and $(\cdot)^\dagger$ are two syntactic stream derivative operators, corresponding to two different total orders, Lemma \ref{PIPPO} implies that $p'(X,\ov\sigma(0),\ov\sigma,\ov\sigma')= p^\dagger(X,\ov\sigma(0),\ov\sigma,\ov\sigma') =(p(X,\ov \sigma))'$, where the last occurrence of $(\cdot)'$ denotes stream derivative.
	
	Ultimately, this coincidence stems from the fact that the asymmetry in the definition of stream derivative of the convolution product, $(\sigma\times \tau)'=\sigma'\times \tau+\sigma(0)\cdot\tau'$, is only apparent. Indeed, taking into account the equality $\sigma(0)=\sigma-X\times\sigma'$, one can obtain the symmetric rule  $(\sigma\times \tau)'=\sigma'\times \tau+\sigma\times\tau'-X\times \sigma'\times \tau'$. Note, however, that the equality $\sigma(0)=\sigma-X\times\sigma'$ cannot be expressed at the syntactic (polynomial) level.
\end{rem}
}

The next lemma is about rational \sde s, and how to convert them into polynomial \sde s.
\revOne{
This result has already appeared in the literature in various forms, see e.g. \cite{BBHR14,Mil10}. Here we keep its formulation as elementary as possible, and tailor it to our purposes.
}
Its proof is a routine application of equation \eqref{eq:inv}.

\begin{lem}[from rational to polynomial \sde s]\label{lemma:rational}
Let $f_i(x,\ov y_0,\ov y)$ for $i=1,...,n$ and $g(x,\ov y_0,\ov y)$  be polynomials, and $\ov r_0\in \K^n$ be such that $g(0,\ov r_0,\ov r_0)\neq 0$. Let $\ov\sigma=(\sigma_1,....,\sigma_n)$ be any tuple of streams satisfying the following system of \revOne{(rational) \sde s\ } and initial conditions: 
\begin{align}\label{eq:rat}
	\revOne{ \sigma'_i}&= \revOne{f_i(X,\ov r_0,\ov \sigma)\cdot g(X,\ov r_0,\ov  \sigma)^{-1} } &&  \revOne{\sigma_i(0)}=r_{0i} && (i=1,...,n)\,.
\end{align}
Then, for a new variable $w$,  there is a polynomial  $h(x,\ov y_0,\ov y,  w)$,   not depending  on $\ov\sigma$,   such that  $(\ov\sigma, \tau)$, with $\tau:= g(x,\ov r_0,\ov \sigma)^{-1}$, is the unique solution of   the following initial value problem of $n+1$  \emph{polynomial} \sde s and initial conditions: 
\begin{align}\label{eq:polratOne}
	\revOne{ \sigma'_i}&= \revOne{f_i(X,\ov r_0,\ov \sigma)\cdot \tau  } & & \revOne{ \sigma_i(0)}=r_{0i} & (i=1,...,n)\\
	\revOne{\tau'}&= \revOne{-g(0,\ov r_0,\ov r_0)^{-1}\cdot  h(X,\ov r_0,\ov \sigma,  \tau) \cdot \tau} && \revOne{ \tau(0)}=g(0,\ov r_0,\ov r_0)^{-1}  \,.\label{eq:polratTwo}
\end{align}
In particular, the polynomial $h(x,\ov y_0,\ov y,  w)$ is obtained from $g' = g'(x,\ov y_0,\ov y,\ov y')$ by replacing each \revOne{variable} $y'_i$ with \revOne{the polynomial} $ f_i(x,\ov y_0,\ov y)\cdot w$, for $i=1,...,n$. Conversely, for  any $(\ov\sigma, \tau)$ satisfying \eqref{eq:polratOne} and \eqref{eq:polratTwo}, we have that $\ov\sigma$ also satisfies \eqref{eq:rat}.
\end{lem}

An important technical ingredient in the proof of the IFT for streams is an operator of \emph{stream partial derivative} $\frac \dh {\dh y_i}$ on polynomials: this will allow us to formulate a stream analog  of the chain rule from calculus\footnote{The chain rule from calculus is: $\frac{\mathrm{d}}{\mathrm{d} x}f(y_1(x),...,y_n(x))= \sum_{i=1}^n   \frac{\partial  }{\partial y_i} f(y_1(x),...,y_n(x))\cdot \frac{\mathrm{d}}{\mathrm{d} x}y_i(x)= \nabla_{\ov y}\ f(y_1(x),...,y_n(x))\cdot (\frac{\mathrm{d}}{\mathrm{d} x}y_1(x),...,\frac{\mathrm{d}}{\mathrm{d} x}y_1(x))^T$.}. \revOne{
For our purposes, a chain rule for  polynomials suffices; for a more general scenario, see \cite[Eq.25]{Rut05}, where composition of streams is introduced (and can be used for covering the case of arbitrary functions).
The following result is instrumental to formally introduce stream partial derivatives and the chain rule for streams.}

\begin{lem}
\label{lemma:chainrule}
For every $p\in \Pol$, any $y'_i\in \ov y'$ can only occur linearly in $p'$, i.e., there is a unique $(n+1)$-tuple  $(q_0,q_1,...,q_n)$ of polynomials in $\K[x,\ov y_0,\ov y]$ such that $p'=q_0+\sum_{i=1}^n q_i\cdot y'_i$.
\end{lem}
\begin{proof}
\revTwo{
	Let us first consider existence.
	We first consider the case in which $p$ is a monomial and we proceed by induction on its total degree.
	The base case follows by the first three cases of Def.~\ref{def:stream_der}:
	for $p=1$, set all $q_i$'s to 0;
	for $p=x$, set $q_0=1$ and all other $q_i$'s to 0;
	for $p=y_i$, set $q_i=1$ and all other $q_j$'s to 0.
	For the inductive case, consider  $p=y_i\cdot m$  with $m\neq 1$ and $y_i=\min(y_i\cdot m)$. By induction hypothesis, there is a unique $(n+1)$-tuple  $(\hat q_0, \hat q_1,...,\hat q_n)$ of polynomials such that $m'=\hat q_0+\sum_{j=1}^n \hat q_j\cdot y'_j$.
	By the fourth case  of Def.~\ref{def:stream_der}, $p':=y'_i\cdot m+y_{0i}\cdot m'$;  then, it suffices to set $q_0=y_{0i}\cdot \hat q_0$, $q_i=y_{0i}\cdot \hat q_i+m$, and all other $q_j$'s to $y_{0i}\cdot \hat q_j$.
	The case when $p$ is a linear combination of monomials follows by  linearity.
	
	As to uniqueness, suppose there are two tuples $(q_0,q_1,...,q_n)$ and  $(\tilde q_0,\tilde q_1,...,\tilde q_n)$ of polynomials in $\K[x,\ov y_0,\ov y]$ such that $p'=q_0+\sum_{i=1}^n q_i\cdot y'_i = \tilde q_0+\sum_{i=1}^n \tilde q_i\cdot y'_i$. This implies $(q_0-\tilde q_0)+\sum_{i=1}^n (q_i-\tilde q_i)\cdot y'_i = 0$. For each $j=1,...,n$, the indeterminate $y'_j$ in the last sum does not occur in any of the terms $(q_i-\tilde q_i)$ ($0\leq i\leq n$), which implies that $(q_j-\tilde q_j)=0$, hence $q_j=\tilde q_j$. This in turn implies $q_0-\tilde q_0=0$, hence  $q_0=\tilde q_0$ as well.
}
\end{proof}

For reasons that will be apparent in a while, we introduce  the following suggestive notation for the   polynomials $q_i$ uniquely determined by $p$ according to Lemma \ref{lemma:chainrule}:
\begin{align*}
\frac{\dh p }{\dh x}&:=q_0 &\frac{\dh p }{\dh y_i}:=q_i\ \ \text{($i=1,...,n$)}&& \nablah_{\ov y}p&:= \left(\frac{\dh p }{\dh y_1},...,\frac{\dh p }{\dh y_n} \right)\,.
\end{align*}
With this notation,  the equality for $p'$ in the lemma can be written in the form of a chain rule:
\begin{align}\label{eq:chr}
p'&:=\frac{\dh p }{\dh x}+(\nablah_{\ov y}p)\cdot {\ov y'}^T\,.
\end{align}
Also, it is easy to check that $\frac{\dh p }{\dh x}\in \Pol$, so that one may write  $\frac{\dh p }{\dh x}(x,\ov y)$ if wanting to emphasize  the dependence on   indeterminates.  \revOne{Practically, $\frac{\dh p }{\dh y_i}$  ($1\leq i\leq n$) can be easily computed from $p'$ by taking its quotient with respect to $y'_i$:  this means  expressing the syntactic stream derivative as $p'=(\cdots) + q\cdot y'_i$, with $y'_i$ not occurring in $(\cdots)$, then letting $\frac{\dh p }{\dh y_i}=q$. Likewise,  $\frac{\dh p }{\dh x}$  can be computed by removing from $p'$ all terms divisible by some $y'_i$. A few  examples are discussed below.

\begin{exa}[partial stream derivatives]\label{ex:psd}
	Let $\alpha$ be a monomial. For $x$ not occurring in $\alpha$, $j\geq 1$ and $y_i\neq x$, we have:
	\begin{itemize}
		\item $\frac{\dh }{\dh x}x^{j} \alpha=x^{j-1} \alpha$,  $\frac{\dh }{\dh x}  \alpha=0$ and $\frac{\dh }{\dh y_i}x^{j} \alpha=0$;
		\item for  $y_i$ not occurring in $\alpha$,  $\frac{\dh }{\dh y_i}  \alpha=0$.
	\end{itemize}
	When $y_i$ occurrs in $\alpha$,   the position of $y_i$ in the total order of variables plays a role in the result (only at a syntactic level, cf. Remark \ref{rem:totalorder}).   As an example,   $\frac{\dh }{\dh y_2}2y_1^2y_2^2y_3=2y_{01}^2y_3(y_2+y_{02})$.

	The partial stream derivative operator is linear.
	As an example, for $p=x^2y_1y_2^3+2y_1y_2^2 +2x +1$, we have: $\frac{\dh p }{\dh x}=xy_1y_2^3  +2$, $\frac{\dh p }{\dh y_1}=  2 y_2^2$ and $\frac{\dh p }{\dh y_2}=2y_{01}(y_{02}+y_2)$.
	These are all instances of a general rule expressed by equation   \eqref{der},  which is established in the proof of Lemma \ref{lemma:equiv}.
\end{exa}
}

{
The following lemma translates the syntactic formula \eqref{eq:chr} in terms of streams. Its proof is an immediate consequence
of \eqref{eq:chr} and of  Lemma \ref{PIPPO}.

\begin{lem}[chain rule for stream derivative]\label{lemma:crsd}
	For any $\ov\sigma$ and $\ov r_0=\ov\sigma(0)$, we have:
	$$(p(X,\ov\sigma))'=\frac{\dh p }{\dh x}(X,\ov r_0,\ov\sigma)+(\nablah_{\ov y}p)(X,\ov r_0,\ov\sigma)\cdot {\ov \sigma'}^T.$$
\end{lem}
}

Now we assume $|\E|=n$, say  $\E=\{p_1,...,p_n\}$.  Fixing some order on its elements, we will sometimes regard $\E$ as a \emph{vector}  of polynomials, and use the  notation    $\E(x,\ov y)$ accordingly. In particular, we let $\nablah_{\ov y}\E$ denote the $n\times n$ matrix of polynomials whose rows are $\nablah_{\ov y}p_i$, for $i=1,...,n$. Evidently, this is the stream analog of the \emph{Jacobian} of $\E$. Moreover, we let $\frac{\dh \E }{\dh x}:=\left(\frac{\dh p_1 }{\dh x},....,\frac{\dh p_n }{\dh x}\right)$. The following lemma is an immediate consequence of \revOne{the  fact that $\E(X,\ov\sigma)'=0$ and of previous lemma, considering $\E$ componentwise}.

\begin{lem}\label{lemma:jacob}
Let $\ov\sigma=(\sigma_1,...,\sigma_n)$ be a solution of $\E$ and $\ov r_0=\ov\sigma(0)$. Then
\begin{align}\label{eq:basic}
	(\nablah_{\ov y}\E)(X,\ov r_0,\ov\sigma)\cdot \ov\sigma'^T+\left(\frac{\dh \E }{\dh x}(X,\ov\sigma)\right)^T&=0.
\end{align}
\end{lem}

{\begin{exa}
	\label{ex:part1}
	Consider the polynomial system made up of a single equation,   $\E=\{p\}$,  where  $p(x,\ov y):=y-(1+xy^2)$ with $\ov y=y=y_1$
	(see also \cite[pg. 117]{Rut05}). We compute
	\begin{align*}
		(\nablah_{\ov y}\E)(x,\ov r_0,\ov y)&=\nablah_y p(x,r_0,y)
		\, =\frac{\dh  p}{\dh y}(x,r_0,y)\,=1\\
		\frac{\dh \E }{\dh x}(x, \ov y)&=\frac{\dh p }{\dh x}(x, y)\,
		=\,-y^2\,.
	\end{align*}
	Hence, for any stream solution $\sigma$ of $\E$, applying Lemma \ref{lemma:jacob}  we get:	
	$$\sigma'-\sigma^2=0\,.$$
	Consider now    $\E=\{q\}$ with $q(x,y):= x^2+y^2-1$, where $y=y_1$ and $y_0=y_{01}$. We compute
	\begin{align*}
		(\nablah_y \E)(x,r_0,y)&=y+y_0\\
		\frac{\dh q }{\dh x}(x,y)&=x\,.
	\end{align*}
	Hence, for any stream solution $\sigma$ of $\E$, applying Lemma \ref{lemma:jacob}  we get:
	$$(\sigma+ \sigma(0))\cdot \sigma' + X=0\,.$$
	%
	%
	%
\end{exa}}

Let us recall a few facts from the theory of matrices and determinants in a commutative ring, applied to the ring $\Sigma$. By definition, a matrix of streams $A\in \Sigma^{n\times n}$ is invertible iff there exists a matrix of streams $B\in \Sigma^{n\times n}$ such that $A\times B=B\times A=I$ (the identity matrix of streams); this $B$, if it exists, is unique and denoted by $A^{-1}$.
It is easy to show that   $A\in \Sigma^{n\times n}$ is invertible if and only if $A(0)\in \K^{n\times n}$ is invertible\footnote{Note this is true only because we insist that the inverse matrix must also lie in $\Sigma^{n\times n}$. Working in the field of formal \emph{Laurent} series, which strictly includes $\Sigma$, this would be   false: e.g. $X(0)=0$, but $X$ has $X^{-1}$ as an inverse.}. By general results on determinants,  $\det(A\times B)=\det(A)\cdot\det(B)$ (Binet's theorem). For streams, this implies that, if $A$ is invertible, then $\det(A)$ as a stream is invertible, that is $\det(A)(0)\neq 0$. Moreover, again by virtue of these general results, the formula for the element of row $i$ and column $j$ of $A^{-1}$ is given by:
\begin{align}
A^{-1}(i,j)&=(-1)^{i+j}\det(A)^{-1}\cdot \det(A_{ji})\label{eq:invA}
\end{align}
where $A_{ji}$ denotes the $(n-1)\times (n-1)$ \emph{adjunct} matrix obtained from $A$ by deleting its $j$-th row and $i$-th column. Also note that, for a $n\times n$ matrix of polynomials, say $P=P(x,\ov y_0,\ov y)$, $\det(P)$ is a polynomial in $x,\ov y_0,\ov y$, and  $\det(P(X,\ov r_0,\ov \sigma))=(\det(P))(X,\ov r_0,\ov \sigma)$.

\begin{thm}[IFT for streams]
\label{th:ift}
Let $\ov r_0\in \K^n$ be such that $\E(0,\ov r_0)=0$ and  $(\nablah_{\ov y}\E)(0,\ov r_0,\ov r_0)$ is invertible as a matrix in $\K^{n\times n}$. Then there is a unique stream solution $\ov\sigma$ of $\E$  such that $\ov\sigma(0)=\ov r_0$. Moreover, $(\nablah_{\ov y}\E)(X,\ov r_0,\ov\sigma)$ is invertible as a matrix in $\Sigma^{n\times n}$ and $\ov\sigma$ satisfies the following system of  $n$ rational \sde s and initial conditions:
\begin{align}\label{eq:jacobSDE}
	\ov\sigma'^T&=-(\nablah_{\ov y}\E)(X,\ov r_0,\ov\sigma)^{-1}\cdot \left(\frac{\dh \E }{\dh x}(X,\ov\sigma)\right)^T &&\ov\sigma(0)=\ov r_0 \,.
\end{align}
Moreover,  from \eqref{eq:jacobSDE} it is possible to build a system  of $n+1$ \emph{polynomial} \sde s in $n+1$ variables and corresponding initial conditions, whose unique solution is $(\ov\sigma,\tau)$, for a suitable $\tau$.
\end{thm}
\begin{proof}
We will first show that the initial value problem   given in \eqref{eq:jacobSDE}   is satisfied by every, if any, stream solution $\ov \sigma$ of $\E$ such that $\ov\sigma(0)=\ov r_0$. 
Indeed, consider any such $\ov\sigma$. As $(\nablah_{\ov y}\E)(X,\ov r_0,\ov\sigma)(0)=(\nablah_{\ov y}\E)(0,\ov r_0,\ov r_0)$, which is invertible by hypothesis, the above considerations on matrix invertibility imply that there exists   $(\nablah_{\ov y}\E)(X,\ov r_0,\ov\sigma)^{-1}$ in $\Sigma^{n\times n}$.  Multiplying equality \eqref{eq:basic} from Lemma \ref{lemma:jacob} to the left by $(\nablah_{\ov y}\E)(X,\ov r_0,\ov\sigma)^{-1}$, we obtain that $\ov\sigma$ satisfies  \eqref{eq:jacobSDE}. Now define the following (matrix of) polynomials:
\begin{itemize}
	\item $g(x,\ov y_0,\ov y):=\det(\nablah_{\ov y}\E)$
	\item $\tilde A:=[\tilde a_{ij}]$ with $\tilde a_{ij}:=(-1)^{i+j}\det((\nablah_{\ov y}\E)_{ji})$
	\item $f_i(x,\ov y_0,\ov y):=-\tilde A_i\cdot \left(\frac{\dh \E }{\dh x}\right)^T$, where $\tilde A_i$ denotes the $i$-th row of   $\tilde A$.
\end{itemize}
Applying our previous observation on the determinant of a matrix of polynomials, we have that $\det((\nablah_{\ov y}\E)(X,\ov r_0,\ov\sigma))=g(X,\ov r_0,\ov\sigma)$, and similarly $(-1)^{i+j}\det((\nablah_{\ov y}\E)(X,\ov r_0,\ov\sigma))_{ji})=\tilde a_{ij}(X,\ov r_0,\ov \sigma)$. Therefore, by the formula for the inverse matrix \eqref{eq:invA}, equation \eqref{eq:jacobSDE} can be written componentwise as follows
\begin{align}\label{eq:ratproof}
	\sigma'_i&= f_i(X,\ov r_0,\ov \sigma)\cdot g (X,\ov r_0,\ov \sigma)^{-1} &&\sigma_i(0)=r_{i0} &&   \text{($i=1,...,n$)}\,.
\end{align}
This is precisely the rational form in \eqref{eq:rat}. 
Then Lemma \ref{lemma:rational} implies that there is a set $\Eqs$ of $n+1$ polynomial \sde s in the indeterminates $\ov y,w$, and corresponding initial conditions $\rho:=(\ov r_0,g(0,\ov r_0,\ov r_0)^{-1})$, satisfied when letting $\ov y,w=\ov \sigma,\tau$, where $\tau=g (X,\ov r_0,\ov \sigma)^{-1}$:
\begin{align}\label{eq:polyOneproof}
	y'_i&= f_i(x,\ov r_0,\ov y)\cdot w  && y_i(0)=r_{i0} && \text{($i=1,...,n$)}\\
	w'&= -g(0,\ov r_0,\ov r_0)^{-1}\cdot  h(x,\ov r_0,\ov y,  w) \cdot w && w(0)=g(0,\ov r_0,\ov r_0)^{-1}  \label{eq:polyTwoproof}
\end{align}
with $h$ obtained from $g$ as described in  Lemma \ref{lemma:rational}. Note the \sde s    $\Eqs$   we have arrived at are purely syntactic and do not depend on the existence of any specific $\ov\sigma$. Now, by Theorem \ref{thm:main} there is a (unique)  solution, say $(\ov\sigma,\tau)$, of the polynomial \sde\ initial value problem $(\Eqs,\rho)$ defined by \eqref{eq:polyOneproof}-\eqref{eq:polyTwoproof}.

We now show that $\ov\sigma$ is a stream solution of $\E$.  By the last part of Lemma \ref{lemma:rational}, $\ov\sigma$ satisfies \eqref{eq:ratproof}, which, as discussed above,   is just another way of writing  \eqref{eq:jacobSDE}. Now  we have
\begin{align*}
	\E(0,\ov\sigma)(0)&=\E(0,\ov r_0)=0\\
	\E(0,\ov\sigma)'&= (\nablah_{\ov y}\E)(X,\ov r_0,\ov\sigma) \,\cdot\, \ov\sigma'^T\;+\; \left(\frac{\dh \E }{\dh x}(X,\ov\sigma)\right)^T\\
	&=  -\,(\nablah_{\ov y}\E)(X,\ov r_0,\ov\sigma)\,\cdot \,  (\nablah_{\ov y}\E)(X,\ov r_0,\ov\sigma)^{-1}\cdot \left(\frac{\dh \E }{\dh x}(X,\ov\sigma)\right)^T \;+\; \left(\frac{\dh \E }{\dh x}(X,\ov\sigma)\right)^T\\
	&=-\left(\frac{\dh \E }{\dh x}(X,\ov\sigma)\right)^T+\left(\frac{\dh \E }{\dh x}(X,\ov\sigma)\right)^T\\
	&=0
\end{align*}
where the second equality is just the chain rule on streams \revOne{(Lemma \ref{lemma:crsd})}, and the third equality follows from \eqref{eq:jacobSDE}. As $\E(0,\ov\sigma)(0)=0$ and $\E(0,\ov\sigma)'=0$, by e.g. the fundamental theorem of the stream calculus \eqref{eq:ftsc} it follows that $\E(0,\ov\sigma)=0$.
This completes the \emph{existence} part of the statement.

As to \emph{uniqueness}, consider any tuple of streams $\ov{ \zeta}\in \Str^n$ that is a stream solution of $\E$ and such that $\ov{\zeta}(0)=\ov r_0$. As shown above, $(\ov\zeta,\xi)$, with $\xi=g(X,\ov r_0,\ov{\zeta})^{-1}$,  satisfies the polynomial \sde\  initial value problem $(\Eqs,\rho)$ defined by \eqref{eq:polyOneproof}-\eqref{eq:polyTwoproof}. By uniqueness of the solution (Theorem \ref{thm:main}), $(\ov{\zeta},\xi)=(\ov\sigma,\tau)$.
\end{proof}

The above theorem  guarantees existence  and uniqueness of a  solution of $\E$, provided that there exists a  unique tuple of  ‘‘initial conditions'' $\ov r_0\in \K^n$ for which $\E$ satisfies the hypotheses of Theorem \ref{th:ift}. The existence and uniqueness of such a $\ov r_0$ must be ascertained by other means. In particular, it is possible that the algebraic conditions $\E(0,\ov r_0)=0$ and  $\det((\nablah_{\ov y}\E)(x,\ov r_0,\ov r_0))\neq 0$ are already sufficient to uniquely determine $\ov r_0$.
There are powerful tools from algebraic geometry that can be applied to this  purpose, such as elimination theory: we refer the interested reader to \cite{Cox} for an introduction.  For now we shall content ourselves with a couple of elementary examples. An extended example is presented in Section \ref{sec:casestudy}.

\begin{exa}
\label{ex:guarded}
\revOne{Let us consider again $\E=\{p\}$  with  $p(x,y):=y-(1+xy^2)$, described in Example \ref{ex:part1}.
	Note that $p(0,r_0)=0$ uniquely identifies  the initial condition $r_0=\sigma(0)=1$. Also note that $\nablah_y p =1$  is invertible at $y=r_0=1$: hence Theorem \ref{th:ift} applies.} The system   \eqref{eq:jacobSDE}  followed by the   transformation of Lemma \ref{lemma:rational} becomes the following polynomial system of \sde s and initial conditions:
\begin{align*}
	y'&=y^2w & y(0)&=1\\
	w'&=0 & w(0)&=1\,.
\end{align*}
Note that the \sde s and initial condition for  $w$ define  the constant stream $1=(1,0,0,...)$, hence the above system can be simplified to the single \sde\ and initial condition: $y'=y^2$ and $y(0)=1$. The unique stream solution of this initial value problem is $\sigma=(1, 1, 2, 5, 14, 42,...)$, the stream of Catalan numbers. Hence $\sigma$ is the only stream solution of $\E$. 

More generally, any set of \emph{guarded} polynomial equations \revOne{\cite{BBHR14}} of the form $\E=\{y_i-(c_i+xp_i)\,:\,i=1,...,n\}$ satisfies the hypotheses   of Theorem \ref{th:ift} precisely when  $\ov r_0=(c_1,...,c_n)$.
Indeed, $\E(0,\ov r_0)=0$, while $\nablah_{\ov y}\E=I$, the $n\times n$ identity matrix, which is clearly invertible. The \sde s  and initial conditions $(\Eqs,\rho)$ determined by the theorem are   given by  $y'_i=p_i$ and $y_i(0)=c_i$ for $i=1,...,n$, plus the trivial   $w'=0$ and $w(0)=1$, that can be omitted.

\revOne{For a  {non} guarded example, consider $\E=\{q\}$ where $q:= x^2+y^2-1$, again  from Example \ref{ex:part1}.  $q(0,r_0)=0$ gives two possible values, $r_0=\pm 1$. Let us fix $r_0=1$. We have $\nablah_y p=y+y_0$, which is $\neq 0$ when evaluated at $y=y_0=r_0$. Applying Theorem \ref{th:ift} and Lemma \ref{lemma:rational} yields the following \sde s and initial conditions:
	\begin{align*}
		y'&=-xw & y(0)&=1\\
		w'&=\frac{xw^2}2 & w(0)&=\frac 1 2\,.
	\end{align*}
	The \sde\ for $w$ arises considering equation \eqref{eq:inv} for the multiplicative inverse of a stream, in detail, letting $w=\frac{1}{y+1}$, we get: $w'=-w(0)\cdot(y+1)'\cdot w= -\frac{1}{2}\cdot (-xw) \cdot w=\frac{xw^2}2$.}

The unique   solution of the derived initial value problem is the stream $\sigma=(1, 0, -1/2, 0,$ $-1/8, 0,-1/16,...)$;  these are the Taylor coefficients of the function $\sqrt{1-x^2}$ around $x=0$. This stream is therefore the unique solution of $\E$ with $r_0=1$. If we fix $r_0=-1$, we obtain $-\sigma$ as the unique solution, as expected.
\end{exa}

\begin{rem}[relation with algebraic series]\label{rem:algebraic}
Recall from \cite{flajo,Stanley} that a stream $\sigma$ is \emph{algebraic} if there exists a nonzero polynomial $p(x,y)$ in the variables $x,y$ such that $p(X,\sigma)=0$.
For   $|\E|>1$, algebraicity of the solution is not in general guaranteed.  \cite[Th.8.7]{flajo} shows that a sufficient condition for algebraicity in this case is that $\E$ be \emph{zero-dimensional}, i.e. that $\E$ has finitely many  solutions when considered as a set of polynomials with coefficients in $\C(x)$, the fraction field of univariate polynomials in $x$ with coefficients in $\C$. In this case, in fact, for each variable $y_i$ one can apply results from elimination theory to get a single nonzero polynomial $p(x,y_i)$ satisfied by $\sigma_i$. See also the discussion in Section \ref{sec:casestudy}.

On the other hand, we do not require zero-dimensionality of $\E$ in Theorem \ref{th:ift}.  Moreover, for the case of polynomials with rational coefficients, \cite[Cor.5.3]{BG21}   observes   that the  unique solution of a polynomial \sde\   initial value problem like \eqref{eq:ivp} is a tuple of   {algebraic} streams.    Then, an immediate corollary of  Theorem \ref{th:ift} is that, under the conditions stated for $\E$ and $\ov r_0$, the unique stream solution of $\E$ is algebraic, even for positive-dimensional systems ---  at least in the case of polynomials with rational  coefficients.
As an example, consider the following system of three polynomials in the variables $x$ and $\ov y= (y_1,y_2,y_3)$:
\begin{align}\label{eq:nonzerod}
	\E= \left\{ \right.& {y_1}{y_3}^{4}+{x}^{2}-{{y_2}}^{2}+{y_2}\; ,\; -{{y_1}}^{2}{y_2}+x{y_3}+{y_1}\;,\\
	&\left. -{{y_1}}^{3}x{{y_3}}^{5}+{{y_1}}^{4}{{y_3}}
	^{4}-{{y_1}}^{2}{x}^{3}{y_3}+{x}^{2}{{y_1}}^{3}+{x}^{2}{y_2}{{y_3}}^{2}-{x}^{2}{{y_3}}^{2}+{{y_1}}
	^{2}-x{y_3}-{y_1} \right\}\,.\nonumber
\end{align}
Considered as a system of polynomials with coefficients in $\C(x)$, $\E$ is not zero-dimensional   --- in fact, its dimension is 1\footnote{\revOne{As checked with Maple's \textsf{IsZeroDimensional} function of the \textsf{Groebner} package.}}. It is readily checked, though, that  for $\ov r_0=(1,1,1)$  we have   $\E(0,\ov r_0)=0$ and $\det((\nablah_{\ov y} \E)(0,\ov r_0, \ov r_0))=12\neq 0$. From Theorem \ref{th:ift}, we   conclude that the unique stream solution $\ov\sigma$ of $\E$ satisfying  $\ov\sigma(0)=\ov r_0$ is algebraic.
\end{rem}

\ifmai
\begin{example}[feedback loop]\label{ex:feedback} Consider the feedback loop example, where $\E$ contains the following polynomials, with $u_0,g_0,h_0\in \reals$.
\begin{align*}
	&u-(u_0+xp_1) & & g-(g_0+xp_2) & & h-(h_0+xp_3) && c - (u - hc)g\,.
\end{align*}
Let $\ov y:=(u,g,h,c)$. Let us assume  at first that   $g_0h_0\neq -1$. Then  $\E(0,\ov r_0)=0$  has a unique solution $\ov r_0=(u_0,g_0,h_0,c_0)$ with $c_0=\frac{g_0u_0}{1+g_0h_0}$. 
The nonsingular stream Jacobian of $\E$ is  and its inverse are
\begin{align*}
	\nablah_{\ov y}\E&=\left[\begin{smallmatrix} 1& 0   & 0 & 0\\0& 1   & 0 & 0\\0& 0   & 1 & 0\\-g & ch - u_0 & cg_0 &  g_0 h_0 + 1\end{smallmatrix}\right] &&
	\nablah_{\ov y}\E^{-1}=\left[\begin{smallmatrix} 1& 0   & 0 & 0\\0& 1   & 0 & 0\\0& 0   & 1 & 0\\ \frac g {g_0 h_0 + 1} & \frac{-c h + u_0}{g_0 h_0 + 1} & \frac{-c g_0}{g_0h_0 + 1},
		& \frac 1 {g_0 h_0 + 1}\end{smallmatrix}\right].
\end{align*}
Now $\frac{\dh \E }{\dh x}=(p_1,p_2,p_3,0)$, hence \eqref{eq:jacobSDE} yields the following system of \sde s and initial condition satisfied for $\ov y=\ov \sigma$ the unique solution:
\begin{align*}
	\ov y'=\left[\begin{smallmatrix}u'\\[2pt] g' \\[2pt] h' \\[2pt] c'\end{smallmatrix}\right]&=-(\nablah_{\ov y}\E)^{-1}\cdot (\frac{\dh \E }{\dh x})^T=\left[\begin{smallmatrix}p_1 \\[2pt] p_2 \\[2pt] p_3 \\[2pt] \frac{c p_3 g_0   + p_2 (c h - u_0) - p_1 g}{g_0 h_0 + 1}\end{smallmatrix}\right] && \ov y(0)=\ov r_0\,.
\end{align*}
This is already in polynomial form.
In   case   $g_0h_0= -1$, we cannot apply Theorem \ref{th:ift}, yet there may be stream solutions of $\E$. To make an example, consider the following concrete instance  of $\E$:
\begin{align*}
	&u  & & g-(1+xg) & & h-(-1+xh) && c - (u - hc)g\,.
\end{align*}
Then $\ov \sigma=(\sigma_1,\sigma_2,\sigma_3,\sigma_4)$, with  $\sigma_1=0$, $\sigma_2=(-1,-1,-1,...)$, $\sigma_3=(1,1,1,...)$ and $\sigma_4=0$, is  a  solution of $\E$, and these streams are trivially definable by a suitable pair  $(\Eqs,\rho)$. However,  $\nablah_{\ov y}\E$ is singular.
\end{example}

\begin{remark}[computing stream partial derivatives]\label{rem:partial}
For the sake of uniform notation,  in what follows let  $y_0$ denote $x$, so we can write a generic monomial in $\Pol$ as $\alpha=y_0^{k_0}\cdots y_n^{k_n}$. Moreover let $y_{00}$ denote $0$, and  $0^0$ denote $1$ in the following. For any $y_i$ and $k\geq 0$, we let $\ccc^k_{y_i} :=\sum_{j=0}^k y_{0i}^j y_i^{k-j}$, with $\ccc^0_{y_i} :=1$ and $\ccc^{-1}_{y_i}:=0$; by our previous convention, we may write $\ccc^k_{y_i}$ also as $\ccc^k_{y_i}(y_{0i},y_i)$. On monomials $\alpha$,   $\frac{\dh\alpha}{\dh y_i}$ for $i=0,...,n$ can be computed as
\begin{align*}
	\frac{\dh\alpha }{\dh y_i}&:= ( y_{00}^{k_0}\cdots y_{0,i-1}^{k_{i-1}}) \, \ccc^{k_i-1}_{y_i} \,(y_{i+1}^{k_{i+1}}\cdots y_{n}^{k_{n}})\,.
\end{align*}
This is extended   by linearity to the whole $\Pol$.
\end{remark}
\fi

\section{Relations with the classical IFT}\label{sec:classical}
We now discuss a relation of our IFT with the classical IFT from calculus; hence, in the rest of the section, we fix $\K=\reals$.
We start with the following lemma.

\begin{lem}\label{lemma:equiv}
Let $p(x,\ov y)$ be a polynomial, $\ov r_0\in \reals^n$, and $y_i$ in $\ov y$. Consider the ordinary  $\frac{\partial p}{\partial y_i}(x,\ov y)$ and stream $\frac{\dh p}{\dh y_i}(x,\ov y_0,\ov y)$ partial derivatives. Then $\frac{\partial p}{\partial y_i}(0,\ov r_0)= \frac{\dh p}{\dh y_i}(0,\ov r_0,\ov r_0)$.
\end{lem}
\begin{proof}
Let $p=x\cdot p_0+q$ where $x$ does not occur in $q$. Write $q$ as a sum of $k$ monomials, $q=\sum_{j=1}^k \alpha_jy^{k_j}_i$, where both $x$ and $y_i$ do  not occur in any of the monomials $\alpha_j$. Moreover, let us write each $\alpha_j$ as $\alpha_j=\beta_j\cdot\gamma_j$, where $\beta_j$ (resp. $\gamma_j$) contains all the $y$'s with index smaller (resp., greater) than $i$.

For the ordinary partial derivative, we have that
\begin{align*}
	\frac{\partial p}{\partial y_i}(x,\ov y)&=x\cdot \frac{\partial p_0}{\partial y_i}+\sum_{j=1}^k k_j\alpha_jy^{k_j-1}_i\,.
\end{align*}
For the stream partial derivative, let us  denote with $\ccc^{h}_{y_i}$ the quantity $\sum_{j=0}^h y_{0i}^j y_i^{h-j}$, with $c^0_{y_i} :=1$ and $c^{-1}_{y_i} := 0$.
Taking into account the rules for $\dh$ and writing $m(\ov u)$ for the evaluation of a monomial $m$ at $\ov y=\ov u$, we have
\begin{align}
	\label{der}
	\frac{\dh p}{\dh y_i}(x,\ov y_0,\ov y)&=\sum_{j=1}^k  \beta_j(\ov y_0)(\ccc^{k_j-1}_{y_i})\gamma_j(\ov y_0)\,.
\end{align}
By denoting with $\ccc^{h}_{y_i}(r_1,r_2)$ the term $\ccc^{h}_{y_i}$ with $r_1\ (\in \reals)$ in place of $y_{0i}$ and $r_2\ (\in \reals)$ in place of $y_i$, we have that
$\ccc^{h}_{y_i}(r,r)=(h+1)\, r^{h}$, for any $r \in \reals$. Upon evaluation of the above polynomials at $x=0$, $\ov y_0=\ov r_0$, $\ov y = \ov r_0$, we get
\begin{align*}
	\frac{\partial p}{\partial y_i}(0,\ov r_0)&=\sum_{j=1}^k k_j\alpha_j(\ov r_0){r_{0i}}^{k_j-1}\\
	\frac{\dh p}{\dh y_i}(0,\ov r_0,\ov r_0)&=\sum_{j=1}^k  \beta_j(\ov r_0)(\ccc^{k_j-1}_{y_i}(r_{0i},r_{0i}))\gamma_j(\ov r_0)\\
	&= \sum_{j=1}^k  \beta_j(\ov r_0) (k_j r_{0i}^{k_j-1})\gamma_j(\ov r_0)\,=\,\sum_{j=1}^k  k_j \alpha_j(\ov r_0)  r_{0i}^{k_j-1}\,.
	\tag*{\qedhere}
\end{align*}
\end{proof}

\newcommand{\A}{\mathcal{A}}
\newcommand{\Tay}{\mathcal{T}}
The above lemma implies that  the classical and stream Jacobian matrices evaluated at $x=0,\ov y=\ov r_0$ are the same: $(\nabla_{\ov y}\ \E)(0,\ov r_0)= (\nablah_{\ov y}\E)(0,\ov r_0,\ov r_0)$. In particular, the first is invertible if and only if the latter is invertible. Therefore, the classical and stream IFT can be applied exactly under the same hypotheses on $\E$ and $\ov r_0$. What is the relationship between the solutions provided by the two theorems? The next theorem precisely characterizes this relationship. In its statement and proof, we make use of the following concept. Consider the set $\A$ of  functions $\reals\rightarrow \reals$ that are real analytic around the origin, i.e., those functions    that admit    a Taylor expansion with a positive radius of convergence around $x=0$. It is well-known that $\A$   forms a $\reals$-algebra. Now consider the function $\Tay$ that sends each $f\in \A$ to the stream $\Tay[f]$ of its Taylor coefficients around 0,   that is $\Tay[f](j)=f^{(j)}(0)/j!$ for each $j\geq 0$.
It is easy to check that $\Tay$ acts as a $\reals$-algebra homomorphism from $\A$ to $\Str$; in particular, by denoting with ‘$\,\cdot\,$' the (pointwise) product of functions, we have that $\Tay[f\cdot g]=\Tay[f]
\convo \Tay[g]$.

\begin{thm}[stream IFT, classical version]\label{th:class}
Let $\ov r_0\in \reals^n$ be such that $\E(0,\ov r_0)=0$ and  $(\nabla_{\ov y}\ \E)(0,\ov r_0)$ is invertible as a matrix in $\reals^{n\times n}$. Then there is a unique stream solution $\ov\sigma$ of $\E$   such that $\ov\sigma(0)=\ov r_0$. In particular, $\ov\sigma=\Tay[f]$, for  $f:\reals\rightarrow \reals^n$ a real analytic function at the origin, which is the unique solution around the origin of 
the following  system of $n$ rational \ode s and initial conditions:
\begin{align}\label{eq:jacobODE}
	\frac{\mathrm{d} }{\mathrm{d} x}f(x)&=-(\nabla_{\ov y}\ \E)(x,f(x))^{-1}\cdot \left(\frac{\partial \E }{\partial x}(x,f(x))\right)^T &&f(0)=\ov r_0 \,.
\end{align}
\end{thm}
\begin{proof}
Under the conditions on $\E$ and $\ov r_0$ stated in the hypotheses,
the classical IFT implies the existence of a unique real analytic function $f:\reals\rightarrow \reals^n$, say $f=(f_1,...,f_n)$, such that $f(0)=\ov r_0$ and $\E(x,f(x))$ is identically 0. Moreover, 
it tells us that $f$ satisfies the system of (polynomial) nonlinear \ode s and initial conditions in \eqref{eq:jacobODE}.
Note that $\det( (\nabla_{\ov y}\ \E)(0,\ov r_0))\neq 0$ and the continuity of $\det(\E(x,f(x)))$ around the origin, guaranteed by the IFT \cite[Th.9.28]{Rudin}, in turn guarantee  that  $(\nabla_{\ov y}\ \E)(x, f(x))$ is nonsingular in a neighborhood of $x=0$.
Let $\ov\sigma=(\sigma_1,...,\sigma_n)$ be the stream of the coefficients of the Taylor series of  $f$  expanded at $x=0$, taken componentwise: $\ov\sigma=\Tay[f]:=(\Tay[f_1],...,\Tay[f_n])$.
Now $\ov\sigma$ is a stream solution of $\E$, as a    consequence of the fact that $\Tay$ is a  $\reals$-algebra homomorphism between $\A$ and $\Str$: indeed, for each $p(x,\ov y)\in \E$, $0=p(x,f(x))$ implies $(0,0,...)=\Tay[0]=\Tay[p(x,f(x))]= p(\Tay[x],\Tay[f])=p(X,\ov\sigma)$.   Uniqueness of $\ov\sigma$ is guaranteed by Theorem \ref{th:ift}, because $(\nablah_{\ov y}\E)(0,\ov r_0)=(\nabla_{\ov y}\ \E)(0,\ov r_0,\ov r_0)$ (see Lemma \ref{lemma:equiv}) and it is invertible by hypothesis.
\end{proof}

A corollary of the above theorem is that one can obtain  the unique stream solution  of $\E$ also by computing the Taylor coefficients of the solution $f$ of  \eqref{eq:jacobODE}. Such coefficients can be computed without having to explicitly solve the system of \ode s. We will elaborate on this point in Section \ref{sec:computational},
\revOne{
where, we will compare in terms of efficiency the method based on \sde s with the method based on \ode s, on   two   nontrivial polynomial systems. Here, we just consider the \ode s method on a simple example.
}

\begin{exa}\label{ex:catalan_class}
Consider again the polynomial system $\E=\{y-(1+xy^2)\}$ in the single variable $y=y_1$, with the initial condition $r_0=1$, seen in Example \ref{ex:guarded}. Since $(\nabla_y\ \E)(x,y)= 1-2xy$ is nonzero at $(0,r_0)$, we can apply Theorem \ref{th:class}. The \ode\  and initial condition in \eqref{eq:jacobODE} in this case are, letting $f=y$: $\frac{\mathrm{d} }{\mathrm{d} x}y(x)=  \frac{y^{2}}{1-2 x y }$ and $y(0)=1$.
This system can be solved explicitly. Alternatively, one can compute the coefficients of the Taylor expansion of the solution, e.g. by successive differentiation: $y(x)=\sum_{j\geq 0}\frac{y^{(j)}(0)}{j!}x^j=1+1x +2x^2 + 5x^3 + 14x^4 +42x^5+\cdots$. Such coefficients form again the stream $\sigma$ of Catalan numbers that is therefore the unique stream solution of $\E$ with $\sigma(0)=r_0=1$.
\end{exa}

\section{An extended example: three-coloured trees}
\label{sec:casestudy}
We consider a polynomial system  $\E$ implicitly defining the generating functions of    `three-coloured trees',       Example 14 in \cite[Sect.4]{flajo}. For each of the three considered colours (variables), \cite[Sect.4]{flajo} shows how to reduce $\E$  to a single nontrivial equation. This implies algebraicity of the series implicitly defined by $\E$: the  reduction is conducted using results from elimination theory \cite{Cox}. 
Here we will  show how to directly transform $\E$ into  a system of polynomial \sde s and initial conditions, $(\Eqs,\rho)$, as implied by the stream IFT (Theorem \ref{th:ift}). As the   coefficients in $\E$ are rational, reduction to \sde s directly implies algebraicity (Remark \ref{rem:algebraic}), besides giving a method of calculating the streams coefficients.   We will also consider   reduction of $\E$ to a system of polynomial \ode s, as implied by the classical version of the IFT (Theorem \ref{th:class}), and compare the obtained \sde\ and \ode\ systems.

Three-coloured trees are binary trees (plane and rooted) with nodes coloured by any of three colours, $a,b,c$, such that any two adjacent nodes have different colours and external nodes are coloured by the $a$-colour.
Let $\An, \B, \Cn$
denote the sets of three-coloured trees with root of the $a,b,c$ color  respectively, and $A,B,C$ the corresponding ordinary generating functions: the $n$-th coefficient of $A$ is the number of  trees with $a$-coloured root and $n$ external nodes; similarly for $B$ and $C$. 
Below, we 
report from \cite[Sect.4,eq.(40)]{flajo} the 
polynomial system    $\E$; to adhere to the notation of Section \ref{sec:background}, we  have replaced the variables $(A,B,C)$ with $\ov y=(y_1,y_2,y_3)$.{}\footnote{We note that  there is a slip in the first equation  appearing in  \cite{flajo}, by which the term $(B+C)^2=(y_2+y_3)^2$ appears with the wrong sign. The correct sign is used here.}
\ifmai
\begin{equation}
\E:\;
\begin{cases}
	\An=&\{\bullet\}+(\B+\Cn)^2\\
	\B=&(\Cn+\An)^2\\
	\Cn=&(\An+\B)^2
\end{cases}	
\Longrightarrow
\begin{cases}
	A-z-(B+C)^2=0\\
	B-(C+A)^2=0\\
	C-(A+B)^2=0
\end{cases}	
\end{equation}
\fi
\begin{equation}\E:\;
	\label{eq:trees}
	\begin{cases}
		y_1-x-(y_2+y_3)^2&=0\\
		y_2-(y_3+y_1)^2&=0\\
		y_3-(y_1+y_2)^2&=0\,.
	\end{cases}	
\end{equation}

\revOne{
	System \eqref{eq:trees} has been derived via the symbolic method  \cite{flajo}, a powerful technique to translate formal definitions of combinatorial objects into equations on generating functions to count those objects.
	For instance,  consider a three-coloured tree   with an $a$-coloured root. It can either be  single node, accounted by $x$ in the first equation, or a root with two subtrees, each with root either of $b$- or of $c$-colour.
	Considering this structure, system  \eqref{eq:trees} can be readily deduced.}

Since the number of external nodes of any empty tree is $0$, we set $\ov r_0=(0,0,0)$. It is immediate to verify that $\E(0, \ov r_0)=0$ and $(\nablah_{\ov y}\E)(0,\ov r_0,\ov r_0) =   \left[\begin{smallmatrix}	1 & 0 &0  \\
	0& 1 & 0  \\
	0& 0 & 1
\end{smallmatrix}\right]=I$, that is obviously invertible, hence Theorem \ref{th:ift} holds, and we generate system \eqref{eq:jacobSDE} in Theorem \ref{th:ift}. In particular, after applying Lemma \ref{lemma:rational}, we get the following polynomial system of \sde s and initial conditions:
{\small \begin{align}
		\label{eq:SDE0}
		\begin{cases}
			y_1'\ =  2wy_1y_2+wy_2y_3-w &  y_1(0)=   0\\
			y_2'\ = -2wy_1^2-4wy_1y_2-w y_1y_3-wy_1-2wy_2y_3 & y_2(0)=   0\\
			y_3'\ = -wy_1y_2-wy_1 - 2wy_2 & y_2(0)=   0\\
			w'\ = 4w^2y_1^2y_2^2+4w^2y_1^2y_2y_3-8w^2y_1^2y_3^2-6w^2y_1^2y_3 + 8 w^2 y_1y_2^3+ & w(0)= -1\\
			\;\;\;\;\;\;\;\;\;\;14w^2y_1y_2^2  y_3 + 6 w^2y_1y_2^2 - 10 w^2y_1y_2y_3^2-8w^2y_1y_2y_3-2w^2\\
			\;\;\;\;\;\;\;\;\;\;y_1 y_2-4w^2y_1y_3^3-7w^2y_1y_3^2 - 7w^2y_1y_3+4w^2y_2^3y_3+6w^2y_2^2\\
			\;\;\;\;\;\;\;\;\;\;y_3^2+3w^2y_2^2y_3-4w^2y_2^2-6w^2 y_2y_3^3-3w^2y_2y_3^2-10w^2y_2y_3-\\
			\;\;\;\;\;\;\;\;\;\;3w^2y_2- 2w^2y_3^2-3w^2y_3\,.
		\end{cases}
	\end{align}
}\noindent
See Appendix \ref{app:details_trees} for details of the derivation. By Theorem \ref{th:ift}, the original polynomial system $\E$ in \eqref{eq:trees} has a unique stream solution $\ov\sigma$ such that $\ov \sigma(0)=\ov r_0$,  and $(\ov{\sigma},\tau)=(\sigma_1,\sigma_2,\sigma_3,\tau)$, for a suitable $\tau$,  is the unique stream solution of \eqref{eq:SDE0}. In particular, we have: $ \sigma_1=(0, 1, 0, 0, 4, 16, 56, 256, 1236,...)$. This matches the generating function  expansion for   $ y_1$  in Example 14 of \cite{flajo}: $ g_1(z)=z+4z^4+16z^5+56z^6+256z^7+1236z^8+...$.

On the other hand, applying the classic IFT (Theorem \ref{th:class}) to system \eqref{eq:trees},  there is a unique real analytic solution $f(x)=\ov y(x)=(y_1(x),y_2(x),y_3(x))$ of the \ode\ initial value problem \eqref{eq:jacobODE} such that $\ov y(0)=\ov r_0=(0,0,0)$. The system in question can be computed starting from the classical Jacobian of $\E$,
$
(\nabla_{\ov y}\ \E)(x,\ov y)= \left[
\begin{smallmatrix} 1&-2 y_2+-2 y_3&-2 y_3-2 y_2\\
	-2 y_1-2 y_3&1&-2 y_3-2 y_1\\
	-2 y_1-2 y_2&-2 y_2-2 y_1&1
\end{smallmatrix}\right].
$
\ifmai
Denoting the determinant of $\nabla_{\ov y}\ \E$ as
{\small
	\begin{align*}
		\tilde{\Delta}:=&-16 y_{1}^2 y_{2}-16 y_{1}^2 y_{3}-4 y_1^2-16 y_1 y_2^2-32 y_1 y_2 y_3-12 y_1 y_2-16 y_1 y_3^2-12 y_1 y_3-16y_2^2\\&y_3-4y_2^2-16 y_2 y_3^2-12 y_2 y_3-4 y_3^3+1,
	\end{align*}
}\noindent

the inverse of the Jacobian matrix $(\nabla_{\ov y}\ \E)(x,\ov y)$ is
$$
\nabla_{\ov y}\ \E^{-1}= {\tilde{\Delta}^{-1}}\cdot
\begin{bmatrix}\substack{4 y_1^2 + 4 y_1 y_2 + 4 y_1 y_3\\ + 4 y_2 y_3 - 1}&\substack{-4 y_1 y_2-4 y_1 y_3-4 y_2^2\\-4 y_2 y_3-2 y_2-2 y_3}&\substack{-4 y_1 y_2-4 y_1 y_3-4 y_2 y_3\\-2 y_2-4 y_3^2-2 y_3}\\ \\
	\substack{-4 y_1^2 - 4 y_1 y_3 - 4 y_1 y_2\\ - 2 y_1 - 4 y_2 y_3 - 2 y_3}&\substack{4 y_1 y_2+4 y_1 y_3+4 y_2^2\\+4 y_2 y_3-1}&\substack{-4 y_1 y_2-4 y_1 y_3-2 y_1\\-4 y_2 y_3-4 y_3^2-2 y_3}\\ \\
	\substack{-4 y_1^2 - 4 y_1 y_2 - 4 y_1 y_3\\ - 2 y_1 - 4 y_2 y_3 - 2 y_2}&\substack{-4 y_1 y_2-4 y_1 y_3-2 y_1-4 y_2^2\\-4 y_2 y_3-2 y_2}&\substack{4 y_1 y_2+4 y_1 y_3+4 y_2 y_3\\+4 y_3^2-1}
	
\end{bmatrix}.
$$
\fi

Since $\frac{\partial \E }{\partial x}(x,\ov y)_{|x=0,\ov y=\ov r_0}=(-1,0,0)$, \eqref{eq:jacobODE} yields  the following system of rational \ode s and initial conditions:
{
	\begin{align}
		\label{eq:ODE}
		\frac{\mathrm{d}}{\mathrm{d} x}\ov y(x)=-(\nabla_{\ov y}\ \E)^{-1}\cdot \left(\frac{\partial \E }{\partial x}\right)^T= {\tilde{\Delta}^{-1}}\cdot\left[\begin{smallmatrix} \substack{4 y_1^2 + 4 y_1 y_2 + 4 y_1 y_3 + 4 y_2 y_3 - 1} \\[2pt]  \substack{-4 y_1^2 - 4 y_1 y_3 - 4 y_1 y_2 - 2 y_1 - 4 y_2 y_3 - 2 y_3} \\[2pt] \substack{-4 y_1^2 - 4 y_1 y_2 - 4 y_1 y_3 - 2 y_1 - 4 y_2 y_3 - 2 y_2} \end{smallmatrix}\right] && \ov y(0)=\ov r_0.
	\end{align}
	where $\tilde{\Delta}:=-16 y_{1}^2 y_{2}-16 y_{1}^2 y_{3}-4 y_1^2-16 y_1 y_2^2-32 y_1 y_2 y_3-12 y_1 y_2-16 y_1 y_3^2-12 y_1 y_3-16y_2^2 y_3-4y_2^2-16 y_2 y_3^2-12 y_2 y_3-4 y_3^3+1$ is the determinant of $\nabla_{\ov y}\ \E$.
	
	Considering a series solution of the system,   we obtain, for the first component of the solution $f$:
	\begin{equation*}
		y_1(x)=x+4x^4+16x^5+56x^6+256x^7+1236x^8+5808x^9+O(x^{10})
	\end{equation*}
	whose coefficients match those of $\sigma_1$ for \eqref{eq:SDE0}. 

	\section{Classical vs. stream IFT: computational aspects}
	\label{sec:computational}
	
	\ifmai
	We focus   on   methods to generate  the streams  implicitly defined by  a polynomial system $\E$ of $n$ nonzero polynomials in $n$ variables $\ov y=(y_1,...,y_n)$, and an initial condition $\ov r_0\in \K^n$ such that $\E(0,\ov r_0)=0$. For the sake of concreteness, in this section we fix $\K=\C$. When $n=1$, we fall in the case of algebraic functions, see Remark \ref{rem:algebraic}. In this case, there exist efficient methods to build a recurrence relation  that generates the stream coefficients, such that the first $k$  stream coefficients can be computed in   $O(k)$ arithmetic operations. Such recurrences   can be automatically generated starting from the polynomial $p(x,y)$ implicitly defining the sequence. As an example, a       recurrence relation for the stream of the Catalan numbers, defined by the \sde\ $\sigma'=\sigma^2$ with $\sigma(0)=1$, is the following: $\sigma(0)=1$, $\sigma(k+1)=\frac{2(2k+1)}{k+2}\sigma(k)$.   These methods are surveyed in e.g. \cite[Ch.4]{flajo}.
	
	When $n>1$, the situation is less  definite. As discussed in Remark \ref{rem:algebraic}, if the polynomial system $\E$ is  {zero-dimensional} when the polynomials are considered with coefficients in $\C(x)$, the above remarks still apply with some modifications. In principle,  it is possible to use elimination theory, in the form of Groebner bases or resultants, to find for each variable $y_i$ a nonzero polynomial $p(x,y_i)$ which lies in the ideal generated by $\E$ and is annihilated by $\sigma_i$ \cite[Ch.4, Th.8.7]{flajo}. However, this process can be computationally very demanding, as building a Groebner basis of $\E$ has a doubly exponential worst case complexity  --- approximately $\Omega(d^{2^n})$, where $d$ is the maximum degree in $\E$. Most important, the situation of zero-dimensionality is rather \emph{exceptional} for  polynomial systems. Indeed, zero-dimensionality of $\E$ is equivalent to asking that  in any   Groebner basis of $\E$, for each variable $y_i$ there is a polynomial  whose leading monomial is of the form $y_i^{k_i}$  for some $k_i\geq 0$ \cite[Ch.5, Sect.3, Th.6]{Cox}. As an example, the polynomial system $\E$ in \eqref{eq:nonzerod} is not zero-dimensional, something that can be readily seen by computing any Groebner basis of it. Yet, the stream solution of $\E$ with the initial condition $(1,1,1)$ \emph{is} algebraic, as argued in Remark \ref{rem:algebraic}.
	
	In cases where applying elimination theory is not efficient  or not possible, one can compute the stream solution   via the \sde\ initial value problem \eqref{eq:jacobSDE} delivered by the   IFT Theorem \ref{th:ift}, provided of course that $(\nablah_{\ov y} \E)(0,\ov r_0,\ov r_0)$ is nonsingular. In full generality,
	\fi
	We   compare the stream and the classical versions of the IFT from a computational point of view.
	First, we   discuss  how the recurrence \eqref{eq:recurrence} can be effectively implemented for \emph{any} polynomial \sde\ initial value problem of the form \eqref{eq:ivp}, not necessarily arising from an application of Theorem \ref{th:ift}.
	The basic idea is to always reduce   products involving more than two factors to binary products, for which the   convolution formula \eqref{eq:sumconv} can be applied. In order to perform  this reduction systematically, let us consider the set $T$ of  all subterms $t=t(x,\ov y)$ that occur   in  the polynomials $p_i$ in $\Eqs$.  We assume that $T$  also includes   all the constants appearing in $\Eqs$, the constant 1, and    all the variables $y_0\,(:=x),y_1,....,y_n$.
	For each  term $t$ in $T$,  a stream $\sigma_t$  is   introduced   via the following recurrence relation that defines   $\sigma_t(k)$. Formally, the definition goes by lexicographic induction on  $(k,t)$, with the second elements ordered according the “\emph{subterm of}'' relation.  For the sake of notation,  below we let $p_0=1$, and let the case $t=c\cdot t_1$  for $c\in \K$  be subsumed by the last clause, where $c$ is treated as the constant stream $(c,0,0,...)$. Finally, $k>0$.
	\begin{equation}\label{eq:reconv}
		\begin{array}{rcll}
			\sigma_t(0)&=&t(0,\ov r_0) \\
			\sigma_{y_i}(k)&=&\sigma_{p_i}(k-1)  & \text{($i=0,...,n$)}\\
			\sigma_{c}(k)&=&0  &  (c\in \K)\\
			\sigma_{t_1+t_2}(k)&= &\sigma_{t_1}(k)+\sigma_{t_2}(k) \\
			\sigma_{t_1\cdot t_2}(k)&= &\sum_{j=0}^{k}\sigma_{t_1}(j)\cdot\sigma_{t_2}(k-j)\,.
		\end{array}
	\end{equation}
	\revOne{
		This   algorithm  for turning a system of \sde s into a system of recurrence relations can be considered as folklore. It has been applied in e.g.    \cite[Sect.10]{HKR17},  to the \sde\  for the Fibonacci numbers, which is linear. Here we explicitly describe it for the general case of polynomial \sde s.
	}
	Its correctness, as stated by the next lemma, is obvious.
	
	\begin{lem}\label{lemma:algo} Let $\ov\sigma=(\sigma_1,...,\sigma_n)$ be the unique stream solution of a problem $(\Eqs,\rho)$ of the form \eqref{eq:ivp}. With the above definition of $\sigma_t$, we have $\sigma_i=\sigma_{y_i}$, for $i=1,...,n$.
	\end{lem}
	
	In a practical implementation of this scheme, one can avoid recurring over the structure of   $t$, as follows.
	At the $k$-th iteration ($k>0$), the values $\sigma_{t}(k)$ are computed and stored by examining  the terms $t\in T$ according to a total order on $T$   compatible with the “\emph{subterm of}'' relation. In this way, whenever  either of the  last two clauses is applied, one can   access the required values   $\sigma_{t_1}(j),\sigma_{t_2}(j)$ up to $j=k$  already computed and stored away in the current iteration. The computation of the $k$-th coefficient $\ov\sigma(k)$, given the previous ones, requires therefore $O(P  k+S)$ multiplications and additions, where $P$ and $S$ are the number of overall occurrences in $T$ of the product and sum operators, respectively. Overall, this means $O(Pk^2+Sk)$ operations for the first $k$ coefficients. This complexity is minimized by choosing a format of   polynomial expressions that minimizes $P$: for example, a Horner scheme (note that Horner schemes exist also for   multivariate polynomials). Memory occupation grows linearly as $O(k(P+S))$.
	
	Another method to generate the coefficients of the  stream solution is applying the classical version of the IFT (Theorem \ref{th:class}), and rely on the \ode\ initial value problem in \eqref{eq:jacobODE}. However, this choice appears to be computationally less convenient. Indeed, apart from the rare cases where \eqref{eq:jacobODE}   can be solved explicitly, one must obtain the coefficients of the solution  by expanding it as a power series --- indeed its Taylor series. Once the rational system \eqref{eq:jacobODE} is reduced to a polynomial form, which is always possible by introducing one extra variable, the coefficients of this power series   can be computed by a recurrence relation similar to that discussed in Lemma \ref{lemma:algo} for \eqref{eq:recurrence}. The catch is that  the  size of the resulting set of terms $T$ is \emph{significantly larger}  for the  \ode\ system  \eqref{eq:jacobODE} than it is for the \sde\ system \eqref{eq:jacobSDE}.
	%
	To understand why, consider that, under the given hypotheses, the \sde\ system  in \eqref{eq:jacobSDE} is equivalent to  $\E'=0$, while the \ode\  system in \eqref{eq:jacobODE} is equivalent to $\frac{\mathrm{d}}{\mathrm{d}x}\E=0$.  Now, the    terms appearing in  $\E'$   are approximately \emph{half the size} of those appearing in $\frac{\mathrm{d}}{\mathrm{d}x}\E$. This is evident  already when comparing with one another   the stream and the ordinary   derivatives of  a bivariate polynomial $p(x,  y)=q_m(  y)x^m+\cdots+q_1(  y)x+ q_0(  y)$:
	\begin{align*}
		p(x,y)'&=\;q_m(  y)x^{m-1}+\cdots+q_1(  y)+  (q_0(  y))'\\
		\frac{\mathrm{d}}{\mathrm{d}x}p(x,y)&=\left(q_m(  y)m x^{m-1}+x^m\frac{\mathrm{d}}{\mathrm{d}x}q_m(  y)\right)+\cdots+\left(q_1(  y)+ x\frac{\mathrm{d}}{\mathrm{d}x} q_1(  y)\right)+\frac{\mathrm{d}}{\mathrm{d}x} q_0(  y)\,.
	\end{align*}
	
	A small experiment conducted with two different systems of polynomials, the  three-coloured trees \eqref{eq:trees} and the one-dimensional system \eqref{eq:nonzerod},  is in agreement with these qualitative considerations. For each of these systems, we have computed a few hundreds  coefficients of the    solution,  using both the methods in turn:  \sde s   via the recurrence relation of Lemma \ref{lemma:algo}  (Theorem \ref{th:ift}), and \ode s via a power series solution (Theorem \ref{th:class}). In the second case, we have used Maple's \textsf{dsolve} command with the \textsf{series} option\footnote{Python and Maple code for this example available at \href{https://github.com/Luisa-unifi/IFT}{https://github.com/Luisa-unifi/IFT}}.  For both systems, we plot in Figure \ref{fig:times} the execution time  as a function of the number of computed coefficients.
	
	\begin{figure}[]
		\centering
		\includegraphics[width = 200pt]{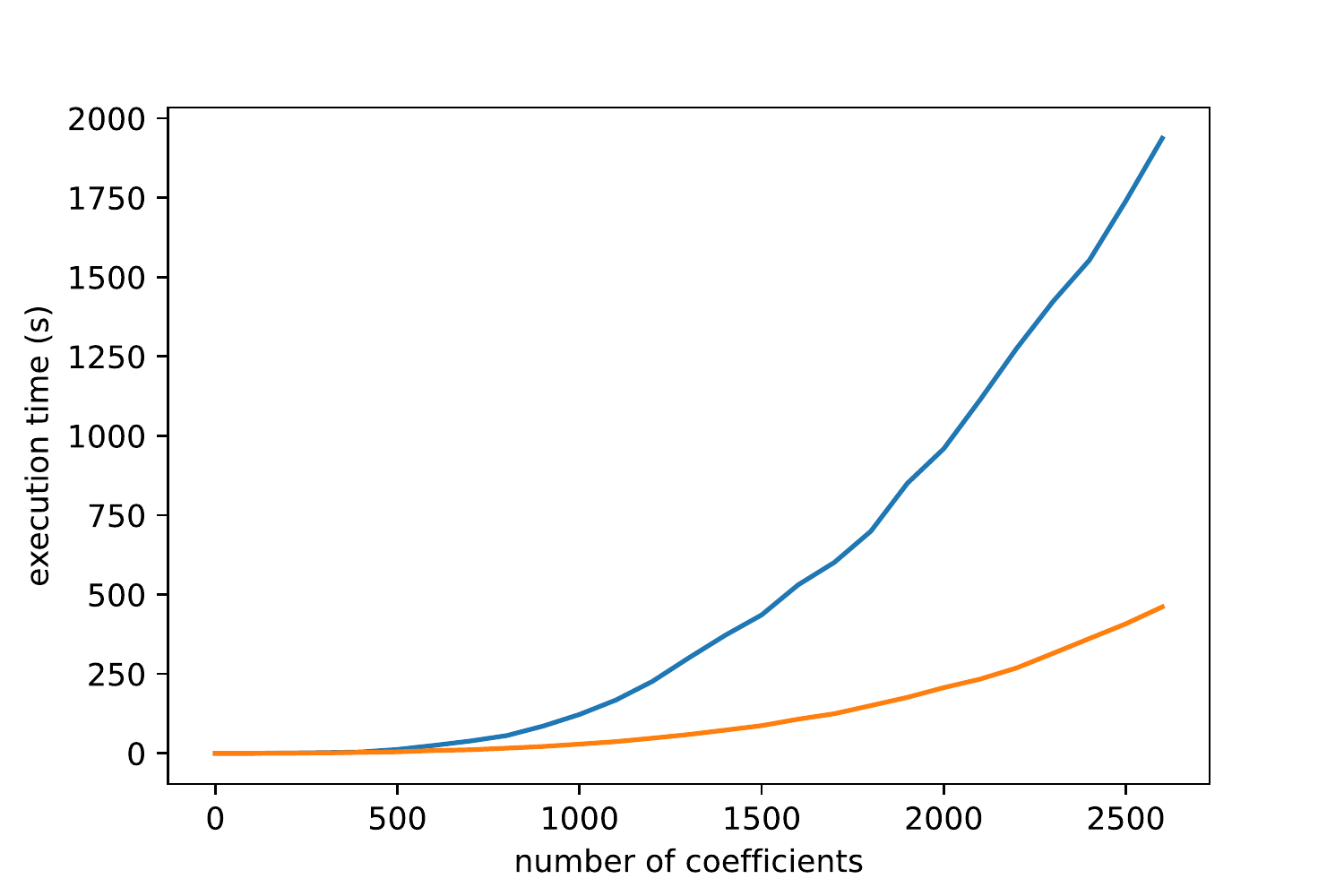}
		\includegraphics[width = 200pt]{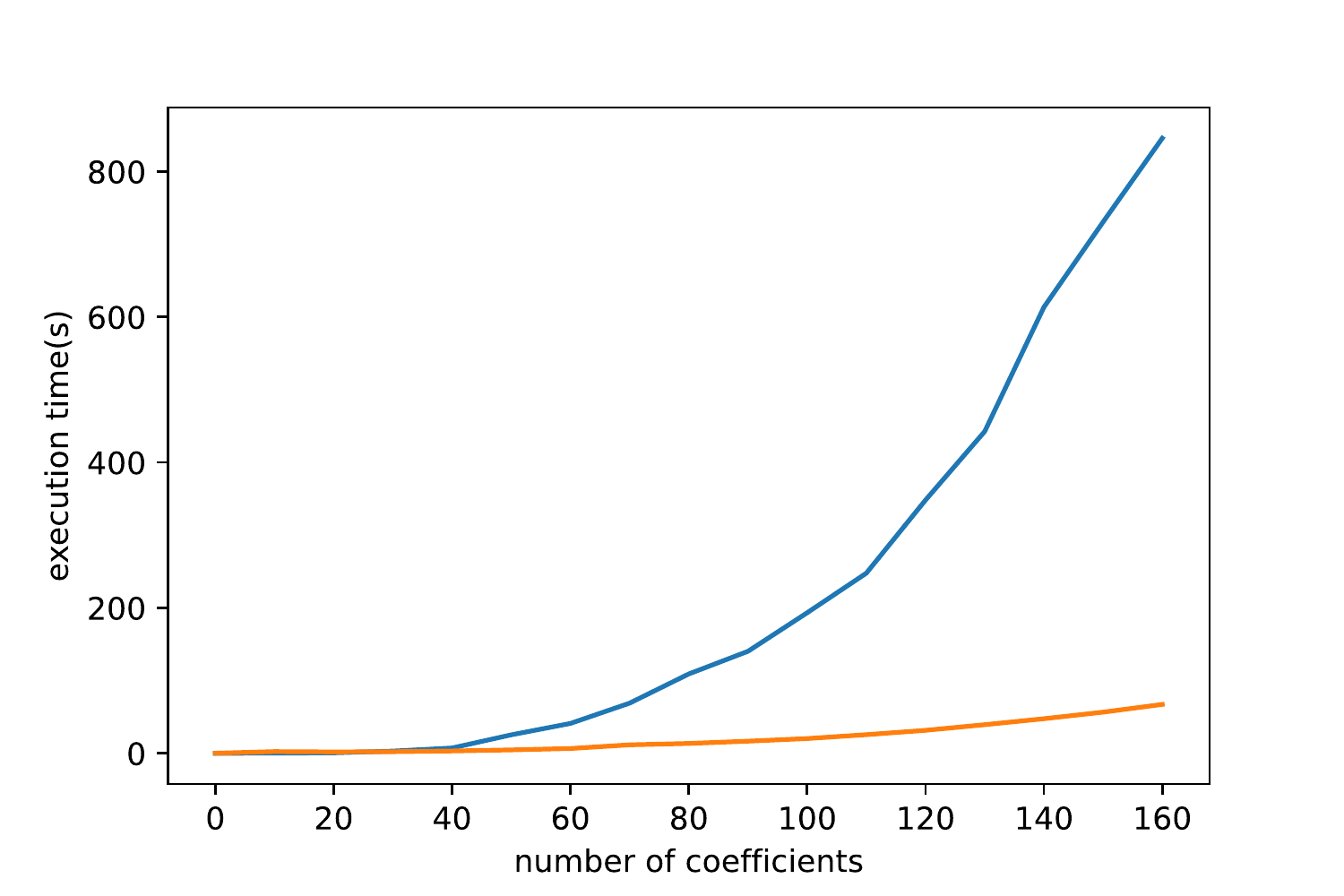}
		\caption{Execution time as a function of the number of computed coefficients  for the stream solution  of system \eqref{eq:trees} (left)  and   of system  \eqref{eq:nonzerod} (right).
			\revOne{The orange (lower) curve is the recurrence relation \eqref{eq:reconv} computed via Lemma \ref{lemma:algo} (authors' Python code); the blue (upper) curve is the power series solution of \eqref{eq:jacobODE} (Maple's \textsf{dsolve})}.
			\label{fig:times}
		}
	\end{figure}
	
	\begin{rem}[Newton method]\label{rem:newton}
		In terms of complexity with respect to   $k$ (number of computed coefficients), \emph{Newton iteration}  applied to formal power series \cite{Lipson,Schon,Galantai,Bostan} does   asymptotically better than the $O(k^2)$ algorithm outlined above.
		In particular, \cite[Th.3.12]{Bostan} shows that, under the same hypotheses of IFT,   the first $k$ coefficients of the solution of a system of algebraic equations can be computed by Newton iteration in time    $O(k\log k)$; on the downside, each iteration of Newton involves in general finding the solution of a $n\times n$ linear system.
	\end{rem}
	
	\ifmai
	\section{Shuffle product}\label{sec:extension}
	In principle, Theorem \ref{th:ift} can be extended to other forms of   product $\pi$ defined on streams, provided: (a) existence and uniqueness of the solution of the polynomial \sde\ problem \eqref{eq:ivp} is ensured; (b) a notion of inverse can be defined, so that the analog of Lemma \ref{lemma:rational} holds; (c)   the stream derivative of $\sigma\,\pi\,\tau$ can be expressed as a linear function of $\sigma'$ and $\tau'$, that is the analog of Lemma \ref{lemma:chainrule} holds.  These conditions hold true, in particular,  for the \emph{shuffle} product $\sigma \otimes\tau$, defined via   binomial convolution: for each $k\geq 0$
	\begin{align}\label{eq:binomial}
		(\sigma\otimes\tau)(k)&:=\sum_{j=0}^j {k \choose {k-j}} \sigma(j)\cdot \tau(k-j)\,.
	\end{align}
	Indeed, as to (a), see e.g. \cite{BG21}. As to (b), we have that the shuffle inverse of $\sigma$, denoted by $\sigma^{-\underline{1}}$, exists provided $\sigma(0)\neq 0$, and one has $(\sigma^{-\underline{1}})'=-\sigma'\otimes \sigma^{-\underline{1}}\otimes\sigma^{-\underline{1}}$, and $(\sigma^{-\underline{1}})(0)=\sigma(0)^{-1}$. And as to (c), we have $(\sigma\otimes \tau)'=\sigma'\otimes \tau+\sigma\otimes \tau'$.
	
	The formula for $(\sigma\otimes\tau)(k)$ in \eqref{eq:binomial} can be given several  interpretations. For instance, if $\sigma$ and $\tau$ express, for each $k$, the number of sequences of length $k$ of two types of objects, say two card decks of different types, then $(\sigma\otimes\tau)(k)$ is the number of sequences of length  $k$ that can be obtained by shuffling together sequences from the first and from the second deck, with the constraint  that  in the resulting deck the order of  each type  is preserved.  See \cite{flajo} for a detailed discussion of this operation. The streams of interest can be  defined implicitly, via a set of (possibly, polynomial) equations $\E$.
	As an example, the stream $\sigma$ counting binary \emph{labelled} rooted trees satisfies $\sigma -\sigma\otimes \sigma-x =0$; cf. \cite{flajo}.
	
	For a polynomial $p(x,\ov y)$, denote by $p_\pi(x,\ov\sigma)$ the stream obtained when evaluating $p$ with $\ov y=\ov\sigma$ and using the product $\pi\in \{\convo,\otimes\}$. We say that $\ov\sigma$ is a  stream $\pi$-solution of $\E$ if $p_\pi(x,\ov\sigma)=0$ for each $p\in \E$, for $\pi\in \{\convo,\otimes\}$. As a matter of fact,  there is direct correspondence between stream $\otimes$- and $\convo$-solutions of a polynomial system,  which makes an explicit extension  of the IFT to the shuffle product unnecessary --- at least from the point of view expressing the coefficients.

	\begin{lem}\label{lemma:shufflevsconv}
		Let $\E\subseteq \K[x,\ov y]$. Then $\ov\sigma$ is a $\otimes$-solution of $\E$ if and only, for some  $\convo$-solution $\ov\tau$ of $\E$, for each $i=1,...,n$  and $k\geq 0$: $\sigma_i(k)=\, k!\,\tau_i(k)$.
	\end{lem}
	
	As an example, the unique stream $\otimes$-solution $\sigma$ of $\E=\{y-(1+xy^2)\}$ is given by $\sigma(k)=k!c_k$, where $\tau=(c_0,c_1,...)$ is the stream of  Catalan numbers.
	In fact, with shuffle product we can go beyond algebraic polynomial equations: we can effectively find a recurrence relations for the solution  of {any } system of \emph{differential} polynomial \ode s, in terms of   binomial coefficients \eqref{eq:binomial}.  The content of the next lemma is the $\otimes$-solutions of \sde\ systems correspond, up to a factor $1/k!$ in the coefficients, precisely to solutions of the corresponding system of \ode s.  Below, we employ the subscript ‘$\otimes$'     to make it explicit in the notation that, when evaluating polynomials  on streams, product must be interpreted as shuffle.
	
	\begin{lem}\label{lemma:odes}
		Let $\mathcal{S}=\{\dot y_i = p_i(x,\ov y)\,:\, i=1,...,n\}$ be a system of polynomial \ode s in the variables $\ov y=(y_1,...,y_n)$, with initial conditions $\rho=\{\ov y(0)=\ov r_0\}$. Let $\ov\sigma$ be the unique stream $\otimes$-solution of the \sde\ i.v.p. $(\Eqs,\rho)$ where $\Eqs=  \{ y'_i = p_{i,\otimes}(x,\ov y)\,:\, i=1,...,n\}$. Then the unique solution $\ov y(x)$ of the i.v.p. $( \mathcal{S},\rho)$ is the power series in $x$: $y_i(x)=\sum_{j\geq 0} \frac{\sigma_i(j)}{j!} x^j$, for $i=1,...,n$.
	\end{lem}
	
	Unfortunately, we do not have for \sde\ systems $\Eqs$ a translation result between $\otimes$ and $\convo$ analogous to that for polynomial systems $\E$ in Lemma \ref{lemma:shufflevsconv}.
	As far as computation of the coefficients is concerned, though, in  analogy with \eqref{eq:recurrence} in Remark \ref{rem:comp}, the stream $\otimes$-solution $\ov\sigma$ of $(\Eqs,\rho)$ satisfies the recurrence ($i=1,...,n$):
	\begin{align}\nonumber
		\sigma_i(0)&=y_i(0)\\
		\sigma_i(k+1)&= p_{i,\otimes}(x,\ov\sigma_{:k})(k)\,.\label{eq:recurrenceBinomial}
	\end{align}
	This recurrence can be  implemented in essentially the same way as discussed for $\convo$ in Section \ref{sec:computational}: in \eqref{eq:reconv}, last clause, one has just to replace (Cauchy) convolution with the binomial convolution \eqref{eq:binomial}.
	
	A consequence of the previous result is that we can also deal with a wide class of  non-polynomial systems of differential equations and equations, provided the involved functions can be expressed by means of polynomial differential equations, which is very often the case.
	
	\begin{exa}[labelled trees]\label{ex:forests}
		Let $t_k$ denote the number of unordered  rooted  labelled trees on an $k$-element set, and let $T(x)=\sum_{j\geq 0}t_j\frac{x^j}{j!}$ be its  {exponential generating function}. Conventionally, $t_0:=1$. One can show combinatorially, see e.g. \cite{Berkley}, that $T(x)$ obeys the functional equation  $T(x)=xe^{T(x)}$. This is not an algebraic polynomial equation, but can be easily transformed into a \emph{differential} polynomial equation. Indeed, by first  differentiating $T$ and then collecting $\dot T$, one gets $\dot T=(1-xe^T)^{-1}\cdot e^T$. Introducing the variables $W=(1-xe^T)^{-1}$ and $E=e^T$, after   further differentiation and some algebra one arrives at the following polynomial system of \ode s and initial conditions in the variables $(T,E,W)$:
		\begin{align*}
			\begin{cases}
				\dot T  \,= WE  & T(0)=0\\
				\dot E  \, =WE^2  & E(0)=1\\
				\dot W  \, =W^2(E+xWE^2)  & W(0)=1\,.
			\end{cases}
		\end{align*}
		Consider the unique stream $\otimes$-solution $\ov\sigma=(\sigma_T,\sigma_E,\sigma_W)$ of the corresponding \sde\ i.v.p.: by Lemma \ref{lemma:odes}, $T(x)=\sum_{j\geq 0}\sigma_T(j)x^j/j!$. Hence the stream of interest is $(t_0,t_1,t_2,...)=\sigma_T$. Applying the recurrence \eqref{eq:recurrenceBinomial}, we can easily compute any number of elements of this stream: $\sigma_T=(1, 1, 2, 6, 24, 120, 720, ...)$, that is $t_j=j^{j-1}$. This is A000169 in \cite{OEIS}.
	\end{exa}
	
	At this point, one might also considering \emph{mixed} specification of polynomial equations and differential equations, that is \emph{DAE} s, Differential-Algebraic Equations. We leave the exploration of this possibility for future work.
	\fi

	\ifmai
	\begin{rem}[relation with exponential g.f.]\label{rem:egf}
		It is   worthwhile to note that, in the case of a single polynomial, if $\sigma$ is a stream $\otimes$-solution of $p(x,y)$, then the \emph{exponential generating function} $eg(z)$ of $\sigma$, defined by $eg(z):=\sum_{j\geq 0} \sigma(j)\frac{z^j}{j!}$, satisfies $p(z,eg(z))=0$ identically; see \cite[Prop.5.1]{BG21}. In other words $eg(z)$ is an algebraic function, and a branch of $p(x,y)$.  This can be extended to the case of a system of equations. Therefore, in cases where all we know about a stream is an implicitly (polynomially) defined  exponential generating function(s), the above simple result concretely means that we can get the stream elements by computing a stream $\convo$-solution of the equations(s). In the above example, $y-(1+xy^2)$ implicitly defines the exponential generating functions of \emph{labelled} binary trees (see e.g. \cite{flajo}); hence the number of such trees with $i$ nodes is given by $i!c_i$.
	\end{rem}
	\fi
	
	\section{Conclusion}\label{sec:concl}
	In this paper, we have presented an implicit function theorem for the stream calculus, a powerful set of tools for reasoning on infinite sequences of elements from a given field. Our theorem is directly inspired from the analogous one from classical calculus. We have shown that  the stream IFT has clear computational advantages over the classical one.
	
	The present work can be extended in several directions.
	\revOne{
		First, we would like to explore the relations of our work with methods proposed in the realm of numerical analysis for efficient generation  of the Taylor coefficients of \ode' solutions; see e.g.  \cite{FW11} and references therein.
	}
	Second, we would like to go beyond the polynomial format, and allow for systems of equations $\E$ involving, for example, functions that are in turn defined via \sde s.
	Finally, we would like to extend the present results to the case of multivariate streams, that is consider  a \emph{vector} $\ov x=(x_1,...,x_m)$ of independent variables, akin to the more general version of the classical IFT. Both extensions seem to pose nontrivial technical challenges.

	\bibliographystyle{alphaurl}
	\bibliography{IFT2}

	\newpage
	\appendix
	\section{Three-coloured trees example:   details}\label{app:details_trees}
	In order to generate the rational system \eqref{eq:jacobSDE}, rather than explicitly determining the inverse of the Jacobian $\nablah_{\ov y}\E$, it is practically convenient firstly to form the equivalent system \eqref{eq:basic} and then solve for  $\ov{y'}$. To this purpose,  we apply the syntactic stream derivative operator to the polynomial equations of system \eqref{eq:trees}, obtaining:
	\begin{align}
		\label{eq:first_linear}
		\begin{cases}
			y_1'-1- y_2'  y_2- y_3'  y_3-2 y_2' y_3&=0\\
			y_2'- y_3'  y_3- y_1'  y_1-2 y_3'  y_1&=0\\
			y_3'- y_1' y_1- y_2' y_2-2 y_1'  y_2&=0\,.
		\end{cases}
	\end{align}
	Then we  note that \eqref{eq:first_linear} is a linear system in the variables $\ov y'=(y_1',y_2',y_3')^T$  of the form $A\ov {y'}=\ov b$, where $A= (\nablah_{\ov y}\E)(0,\ov r_0,\ov y) =   \left[\begin{smallmatrix}	1 & -y_2-2y_3 &-y_3  \\
		-y_1 & 1 & -y_3-2y_1  \\
		-y_1-2y_2& -y_2 & 1
	\end{smallmatrix}
	\right]$ and $\ov b=-(\nablah_x\E)(0,\ov y)^T=(1,0,0)^T$. Note that this is another way of writing system \eqref{eq:basic}. 
	Now, we solve $A\ov {y'}=\ov b$ for $\ov {y'}$. Denoting the determinant of $\nablah_{\ov y}\E$ as {\small
		$$ \Delta=2  y_1^2 y_2 + 4 y_1^2 y_3 + 4  y_1 y_2^2 + 10 y_1  y_2  y_3 + 3 y_1  y_2 + 2  y_1  y_3^2 + 3  y_1  y_3 + 2  y_2^2  y_3 + 4  y_2  y_3^2 + 3 y_2  y_3 - 1$$
	}\noindent
	and taking into account the initial condition, we arrive at \eqref{eq:jacobSDE}:
	\begin{align}\label{eq:rationalSDE}
		\begin{cases}
			y_1' \, = {\Delta^{-1}}\cdot (2  y_1 y_2 +  y_2  y_3 - 1)   &   y_1(0)=0\\
			y_2' \,=  {\Delta^{-1}}\cdot (-2  y_1^2 - 4  y_1 y_2 -  y_1  y_3 -  y_1 - 2 y_2  y_3 ) \;\; &  y_2(0)=0\\
			y_3' \, = {\Delta^{-1}}\cdot(- y_1  y_2 -  y_1 - 2  y_2)     & y_3(0)=0\,.
		\end{cases} 
	\end{align}
	In order to convert the above rational \sde\ initial value problem to a polynomial one,  we can apply Lemma \ref{lemma:rational}.
	In practice,  we replace $\Delta^{-1}$ with a new variable $w$, then we add the corresponding equation   to system \eqref{eq:rationalSDE}. In order to derive   a \sde\ for the variable $w$, we recall that    the multiplicative inverse $\sigma^{-1}$ of a stream $\sigma$   such that $\sigma(0)\neq 0$  satisfies the \sde\ and initial condition  \eqref{eq:inv}.
	In our case, starting from $w' =w(0) \Delta'w$, and calling $p_1,p_2,p_3$ the right-hand sides of the \sde s in \eqref{eq:rationalSDE}, we get
	{\small
		\begin{align*}
			\label{eq:SDE}
			\begin{cases}
				w'  &=  w(0)\cdot (2y_1'y_1y_2+4y_1'y_1y_3+4y_1'y_2^2+10 y_1'y_2 y_3+3y_1'y_2+2  y_1'y_3^2+3y_1' + y_3\\
				&\;\;\;\; +2y_2'y_2y_3+4y_2'y_3^2+3y_2'y_3)w\\
				&= -(2p_1y_1y_2+4p_1y_1y_3+4p_1y_2^2+10p_1y_2 y_3+3p_1y_2+2  p_1y_3^2+3p_1 + y_3\\
				&\;\;\;\; +2p_2y_2y_3+4p_2y_3^2+3p_2y_3)w \\
				w(0)&= -1
			\end{cases}
		\end{align*}
	}\noindent
	where the initial condition $w(0)=-1$  is implied by   $\Delta(0, \ov r_0)=-1$. Finally, expanding the $p_i$'s and putting everything together, we obtain the following polynomial system of \sde s and initial conditions:
	{\small \begin{align}
			\begin{cases}
				y_1'\ =  2wy_1y_2+wy_2y_3-w &  y_1(0)=   0\\
				y_2'\ = -2wy_1^2-4wy_1y_2-w y_1y_3-wy_1-2wy_2y_3 & y_2(0)=   0\\
				y_3'\ = -wy_1y_2-wy_1 - 2wy_2 & y_2(0)=   0\\
				w'\ = 4w^2y_1^2y_2^2+4w^2y_1^2y_2y_3-8w^2y_1^2y_3^2-6w^2y_1^2y_3 + 8 w^2 y_1y_2^3+ & w(0)= -1\\
				\;\;\;\;\;\;\;\;\;\;14w^2y_1y_2^2  y_3 + 6 w^2y_1y_2^2 - 10 w^2y_1y_2y_3^2-8w^2y_1y_2y_3-2w^2\\
				\;\;\;\;\;\;\;\;\;\;y_1 y_2-4w^2y_1y_3^3-7w^2y_1y_3^2 - 7w^2y_1y_3+4w^2y_2^3y_3+6w^2y_2^2\\
				\;\;\;\;\;\;\;\;\;\;y_3^2+3w^2y_2^2y_3-4w^2y_2^2-6w^2 y_2y_3^3-3w^2y_2y_3^2-10w^2y_2y_3-\\
				\;\;\;\;\;\;\;\;\;\;3w^2y_2- 2w^2y_3^2-3w^2y_3\,.
			\end{cases}
		\end{align}
	}\noindent

	\ifmai
	\vsp
	Let us introduce some notation for the next lemma. The \emph{Hadamard} stream product is defined by $(\sigma\had\tau)(i):=\sigma(i)\cdot\tau(i)$ for each $i\geq 0$. Let $\fact:=(0!,1!,2!,...,i!,...)$ be the stream of factorial numbers\footnote{In a generic field of characteristic 0, one lets $i$ denote $1+1+\cdots+1$ ($i$ times).} and  $\fact^{-1}:=(1/0!,1/1!,1/2!,...,1/i!,...)$ its Hadamard inverse, that is $\fact\had\fact^{-1}=(1,1,1,...)$, which is the identity of the Hadamard product. One can easily check by elementary calculations that  for any two streams $\sigma$ and $\tau$
	\begin{align*}
		\fact^{-1}\had(\sigma\otimes\tau)&=(\fact^{-1}\had \sigma)\convo(\fact^{-1}\had \tau)
	\end{align*}
	(see also \cite{BG21}). Given that $\had$ distributes over sum, the above correspondence extends to polynomials. More precisely, with the   notation $\fact^{-1}\had\ov\sigma:=(\fact^{-1}\had \sigma_1,...,\fact^{-1}\had\sigma_n)$, one has:
	\begin{align*}
		\fact^{-1}\had p_\otimes(x,\ov\sigma)&=p_\convo(x,\fact^{-1}\had\ov\sigma)\,.
	\end{align*}
	
	\vsp
	\begin{proof_of}{Lemma \ref{lemma:shufflevsconv}}
		Let $\ov\sigma$ be a $\otimes$-solution of $\E$. Note that $p_\otimes(x,\ov\sigma)=0$ iff $\fact^{-1}\had p_\otimes(x,\ov\sigma)=0$. But $\fact^{-1}\had p_\otimes(x,\ov\sigma)=  p_\convo(\fact^{-1}\had\ov\sigma)$, hence, letting $\tau=\fact^{-1}\had\ov\sigma$, we have that $\ov\sigma=\fact\had\ov\tau$ for $p_\convo(x,\ov\tau)=0$: this proves one implication. The opposite implication is proven similarly.
	\end{proof_of}
	
	\vsp
	\begin{proof_of}{Lemma \ref{lemma:odes}}
		Based on the right-hand sides  polynomials in $\mathcal{S}$ (and $\Eqs$), $p_1,...,p_n$, one defines syntactically a classical  --- or \emph{Lie} ---   derivative $\Lie(\cdot)$ on polynomials: the main clause is $\Lie(y_i m)=p_i m+y_i\Lie(m)$, the other clauses are obvious. Note that $\Lie(p) \in \reals[x,\ov y]$. We note the following facts that connect Lie, ordinary and stream derivatives, and that hold true for any polynomial $p(x,\ov y)$.
		\begin{itemize}
			\item[(a)]
			$p(x,\ov\tau)'=\Lie(p)(x,\ov\tau)$, for any tuple of streams $\ov\tau$. On the left-hand side,   $(\cdot)'$ denotes, as usual,    {stream} derivative. Both on the left and on the right-hand it is understood of course that the evaluation of  the polynomials is carried out by interpreting product as $\otimes$.  For $\ov\tau=\ov\sigma$, the stream $\otimes$-solution of $(\Eqs,\rho)$, this yields: $p(x,\ov\sigma)'(0)=\Lie(p)(0,\ov r_0)$.
			\item[(b)]  $\Lie(p)(0,\ov r_0)=\frac{\mathrm{d}}{\mathrm{d} x}p(x,\ov y(x))_{|x=0}$, where $\frac{\mathrm{d}}{\mathrm{d} x}$ denotes the ordinary derivative  with respect to $x$.  This is a well-known general  fact   about Lie derivatives.
		\end{itemize}
		Facts (a) and (b) generalize to the $k$th derivatives $(\cdot)^{(k)}$,  for $k\geq 0$,   inductively as expected, with the convention that $e^{(0)}=e$ for any $e$.  As a consequence, for each $k\geq 0$ we have:
		\begin{align*}
			p(x,\ov\sigma)(k)& =p(x,\ov\sigma)^{(k)}(0)  \;  =\Lie^{(k)}(p)(0,\ov r_0)\\
			& =\frac{\mathrm{d}^{(k)}}{\mathrm{d} x^k}\,p(x,\ov y(x))_{|x=0}
		\end{align*}
		where the second equality follows from (a) and the last one from (b). Now, instantiating the previous equality to $p=y_i$, for any $i=1,...,n$, we get: $\sigma_i(k)=\frac{\mathrm{d}^{(k)}}{\mathrm{d} x^k}\,y_i(x)_{|x=0}$. Therefore, $\sigma_i(k)/k!$ is the $k$th coefficient of the Taylor expansion of $y_i(x)$ from $x=0$.
	\end{proof_of}
	\fi

	\end{document}